\documentclass{article}

\usepackage{algorithm,algorithmic,procedure,graphicx,latexsym,amsfonts,amsmath,amssymb,amsmath}
\usepackage{color}
\usepackage[latin1]{inputenc}

\usepackage[english]{babel}

\newcommand{\rcp}{\longleftarrow}
\newtheorem{property}{Property}
\newtheorem{definition}{Definition}
\newtheorem{lemma}{Lemma}
\newtheorem{corollary}{Corollary}
\newtheorem{proposition}{Proposition}
\newtheorem{remark}{Remark}
\newtheorem{theorem}{Theorem}
\newtheorem{example}{Example}

\newenvironment{proof}{\noindent{\bf Proof}}{\hfill$\square$}

\def\restriction#1#2{\mathchoice
              {\setbox1\hbox{${\displaystyle #1}_{\scriptstyle #2}$}
              \restrictionaux{#1}{#2}}
              {\setbox1\hbox{${\textstyle #1}_{\scriptstyle #2}$}
              \restrictionaux{#1}{#2}}
              {\setbox1\hbox{${\scriptstyle #1}_{\scriptscriptstyle #2}$}
              \restrictionaux{#1}{#2}}
              {\setbox1\hbox{${\scriptscriptstyle #1}_{\scriptscriptstyle #2}$}
              \restrictionaux{#1}{#2}}}
\def\restrictionaux#1#2{{#1\,\smash{\vrule height .8\ht1 depth .85\dp1}}_{\,#2}} 

\title{Universal adaptive self-stabilizing traversal scheme: random walk and reloading wave}
\author{Thibault Bernard \and Alain Bui \and Devan Sohier}
\newcommand{\C}{\cal C}
\newcommand{\LC}{\cal LC} 

\begin{document}

\maketitle

\begin{abstract}
In this paper, we investigate random walk based token circulation in dynamic environments subject to failures. We describe hypotheses on the dynamic environment that allow random walks to meet the important property that the token visits any node infinitely often. The randomness of this scheme allows it to work on any topology, and require no adaptation after a topological change, which is a desirable property for applications to dynamic systems. For random walks to be a traversal scheme and to answer the concurrence problem, one needs to guarantee that exactly one token circulates in the system. In the presence of transient failures, configurations with multiple tokens or with no token can occur. The meeting property of random walks solves the cases with multiple tokens. The reloading wave mechanism we propose, together with timeouts, allows to detect and solve cases with no token. This traversal scheme is self-stabilizing, and universal, meaning that it needs no assumption on the system topology. We describe conditions on the dynamicity (with a local detection criterion) under which the algorithm is tolerant to dynamic reconfigurations. We conclude by a study on the time between two visits of the token to a node, which we use to tune the parameters of the reloading wave mechanism according to some system characteristics.
\end{abstract}

\section{Introduction}
Concurrence control is one of the most important requirements in distributed systems and have been investigated for 40 years. 
The emergence of peer-to-peer networks, of wireless mobile networks has renewed the context of the design of protocols used in distributed applications. These networks require a new modeling and new solutions to take into account their intrinsic dynamicity. 

In this paper, we focus on token circulation based solutions: the concurrent access to the shared resource is managed by a \emph{``token''} message that circulates in the distributed system.
We present  a self-stabilizing universal traversal scheme based on a random walk, with a particular focus on dynamic systems.

In distributed computing, a random walk is implemented by a \emph{Token} message that is sent from node to node in a random fashion: each time a node receives a \emph{Token} message, it executes a code that only the token owner is allowed to execute, and then forwards the token to one of its neighbor chosen at random.

Properties of random walks allow to design a traversal scheme using only local information \cite{AKLL+79}: such a scheme is not designed for one particular topology and need no adaptation to fit other ones. Moreover, random walks offer the interesting property to adapt to the insertion or deletion of nodes or links  in the network without modifying any of the functioning rules. With the increasing dynamicity of networks, these features are becoming crucial: redesigning a new browsing scheme at each modification of the topology is impossible, and flooding-based solutions can lead to the congestion of the network. 

An important result of this paradigm is that  the token will eventually visit (with probability 1) all the nodes of a system, even if it is impossible to capture an upper bound on the time required to visit all the nodes of the system.

Random walks based traversal schemes have be used in many theoretical distributed computing problems: mutual exclusion \cite{IsJa90}, spanning tree construction \cite{BaZe89}, or at applicative level: decentralized recommender system \cite{KLMT10} and concurrence management in Grid computing \cite{Ciuf10}. 

The random walk traversal scheme can be affected by different errors. In this paper, we manage them in a \emph{self-stabilizing} fashion, as introduced by Dijkstra in \cite{Dijk74}. After a fault, a self-stabilizing system is led to an arbitrary configuration but eventually recovers a normal behavior and then satisfies the specification of the problem.

\subsection*{Related works}

The token circulation can be affected by only 2 errors:
\begin{itemize}
\item the absence of tokens;
\item the presence of more than one token.
\end{itemize}

Both faults are violations of global properties of the system. However, the second fault may entail (and in our algorthim, will eventually entail) the local property that a node holds several tokens at once. Then this node can remove all of them but one, which leads, when all duplicate tokens are removed, to a correct configuration. The first fault has no locally checkable certificate, so that a global mechanism (meaning a mechanism involving all nodes) has to be put in place. This fault is of \emph{communication deadlock} nature: all nodes are waiting for messages and there are no messages on the communication links. The solution proposed by \cite{GoMu91}  is to use timeout: when a node has not seen the token for a long time, it creates a new one. In \cite{Varg00}, the author proposes a message passing adaptation of Dijkstra algorithm \cite{Dijk74}. In particular, a self-stabilizing token circulation algorithm on an undirected ring is presented. Communication deadlock is solved by a timeout process in a distinguished node called the root. Nevertheless, duplication of tokens may occur. 

To solve this problem, the author introduces the counter flushing paradigm and designs a self-stabilizing token circulation algorithm. The idea of counter flushing is used in numerous  papers dealing with self-stabilization in message passing model, as in \cite{ChWe05,HaVi01}. 
This idea is based on a bound on the time between two successive receptions of the token, which we cannot have with a random walk. 
Starting from a configuration in which there is a single token, eventually a token is created unnecessarily, which violates the specification.

In \cite{DoSW06}, the authors use  random circulating tokens (they call agents) to broadcast information in communication groups. To cope with the situation where no agent exists in the system, the authors use a timer based on the cover time of an agent ($k\times n^3$). They precise as a concluding remark \textit{``The requirements will hold with higher probability if we enlarge the parameter $k$ for ensuring the cover time$[$\ldots$]$''}. In our case the obtention of a single token is a strong requirement, and the use of a parameter $k$ which increases the probability to reach a legitimate configuration cannot be used.

Works have been led on the random walk token circulation paradigm (see \cite{Coop11}), in particular to reduce the average time between two successive visits by the token or to attain a given stationary distribution of the token locations (\cite{IKOY02, NOSY10}).

\subsection*{Contribution} 

Our random walk based solution is self-stabilizing: it tolerates transient failures. If there is no token in the system, upon timeout the missing token is recreated. Our solution is decentralized (no distinguished node) and only the expected time for a random walk based token to cover the system can be captured. Each node is candidate to regenerate the token, and even the choice of an arbitrary timeout period implies that the system could never stabilize: an infinite production of unnecessary tokens can occur. 

No node can ascertain that the token does not exist, due to the way the token moves. However, the longer a node has not seen the token, the more probable it is that no token exists. Thus, upon a timeout, a node should create a new token. To avoid this creation in cases when a token already exists, tokens periodically inform nodes of their existence, which inhibits tokens creation. We call this process a \emph{reloading wave}. Each node that has been previously visited by this token receives the reloading wave, resets its timer  and is thus forbidden to create a duplicated token for the next period. The only case when a node creates a new token after the timeout period, corresponds to  a situation in which the node has never been visited by the token.

The reloading wave information should be broadcast efficiently and reliably through the network. The reloading wave is defined in connection with the token. We only use the information collected and stored by the token through its traversal. Thus there is no additional protocol.  In such a token, called a circulating word,  a dynamic self-stabilizing tree is maintained, through which the information is broadcast.  

The reloading wave propagation is periodical. The tree used to broadcast the wave is adaptive: it evolves with the moves of the token. Thus, two propagations of the reloading wave will likely use different propagation trees.



\subsection*{Outline}
In section 2, we present  our model of distributed system and some preliminary notions about random walks and self-stabilizing systems. In section 3, we propose a token circulation scheme in a dynamic environment. We prove that this scheme guarantees the specification of token circulation as long as topology changes are independent of the token moves. In section 4 we introduce the reloading wave mechanism to design a self-stabilizing version of the previous algorithm. This new mechanism is proved to work in a static environment. The case of a dynamic environment is discussed in section 5. A criterion on the mobility pattern to make the algorithm robust against topological reconfigurations is determined. In the last section, we propose to optimize a parameter (timeout) of our algorithm to accelerate the convergence of our algorithm to a configuration where the specification of the problem is satisfied.

\section{Model and Preliminaries}

\subsection{Distributed systems} \label{MDS}
We consider a distributed system as an undirected connected graph $G =(V,E)$, where $V$ is a set of nodes with $|V| = n$ and $E$ is the set of bidirectional communication links with $|E|= m$. A node is composed of a computing unit and a message queue. A communication link $(i,j)$ exists if and only if $i$ and $j$ are neighbors. Every node $i$  maintains a set of its neighbors ids (denoted by $N_i$). The degree of $i$ is the number of neighbors of $i$, i.e. $|N_i|$  (denoted by $\deg(i)$). We consider a distributed system in which all nodes have distinct identities. We assume an upper bound $\mathcal{N}$ on the number of nodes in the network, an upper bound on the delay to deliver a message and an upper bound on the processing time for each node. The sum of these two bounds corresponds to the time for receiving and treating a message. In the sequel, we take this as time unit. Moreover, we assume reliable channels during and after the stabilization phase.

\subsection{Model}

A configuration of  the system is an instance of the nodes states and a multi-set of messages in transit in the links. $Token_\gamma $ is the set of all token messages in transit in the network at the configuration $\gamma$, and $Token_{\gamma}(i)$ is the set of the Token messages heading toward node $i$ at configuration $\gamma$.
A computation $e$ of the system is a sequence of configurations $\gamma_1, \gamma_2, \dots , \gamma_k, \ldots$ such that the configuration $\gamma_{k+1}$ is reached from $\gamma_k$ denoted by ($\gamma_{k} \rightarrow \gamma_{k+1}$) by a single step, a step being an atomic process of one message in the system.  
A configuration $\delta$ is said reachable from $\gamma$ and denoted by $\gamma \rightarrow ^* \delta$ if  there exists a sequence such that
$\gamma = \gamma_0 \rightarrow  \gamma_1  \rightarrow \ldots \gamma_{k-1}  \rightarrow \gamma_k = \delta$. 
Let $\C$ be the set of possible configurations of the system and ${\cal E}$ be the set of all possible computations of the system. The set of computations starting with the  configuration $\gamma$  is 
denoted by ${\cal E}_\gamma$. The set of computations of ${\cal E}$ whose initial configurations are all elements of $A \subset \C$ is denoted by ${\cal E}_A=\bigcup_{\gamma\in A}\mathcal E_\gamma$.

The only nodes variables in our algorithm is a timeout.

\begin{remark}
Since the algorithm we design is random, it would be more accurate to describe a computation as a random process $\mathcal E_\gamma(\omega)=(\gamma_1(\omega), \gamma_2(\omega), \ldots)$, with $\gamma_i:\omega\in\Omega\rightarrow\gamma_i(\omega)$ random variables. Then, the random choice of a neighbor to which the token is sent would make a random walk of the sequence of vertices to which a given token is sent, which is enough to establish the properties required to prove the algorithm. Thus, to avoid overly unwieldy notations, we skip the $\omega$ in the sequel and explicitly use the relevant properties of random walks when required.
\end{remark}

\subsection{Failures and self-stabilization} 

A transient fault is a fault that causes the state of a process (its local state, program counter, and variables) and of a channel (arbitrary messages may be removed and added) to change arbitrarily without further affecting the behavior of the algorithm. An algorithm is called self-stabilizing if it is resilient to transient failures in the sense that, when initiated in an arbitrary system configuration, and no other transient faults occur, the algorithm converges to a legitimate configuration after which it performs its task correctly (see \cite{Dijk74,Dole00}). Thus, a self-stabilizing system experiencing any transient failure, eventually recovers its normal behavior.

As we work with random walks, we cannot ascertain the time at which a property will be true, but we can know with high probability that it will be. \emph{``With high probability''} (in the sequel, \emph{``whp''}) means that the probability that this event never occurs is zero, in the sense that nothing forbids that this event does not occur (one can find an infinite execution without the occurrence of this event), but as times goes by, it is less and less likely that the event has not occurred. Thus, most of the properties we will prove are \emph{whp}, and the convergence times will be expected times (no deterministic bound can be provided).

$\mathcal{C}$ being the set of all configurations of the system, an algorithm is self-stabilizing if there is a set of legitimate configurations $\mathcal{LC}$ such as: \begin{enumerate}
\item the system eventually reaches a legitimate configuration (\emph{convergence} property);
\item starting from any legitimate configuration, the system remains in $\mathcal{LC}$ (\emph{closure} property);
\item starting from any legitimate configuration, the execution of the algorithm verifies the specifications of the problem (\emph{correctness} property).
\end{enumerate}

More formally, in this paper we use the notion of \emph{attractor} to define the self-stabilization concept.

\begin{definition}[Attractor] \label{attractor}
Let $B\subset \C$ and $A\subset B$, $A$ is an attractor of $B$ if and only if: 
\begin{itemize}
\item \textbf{convergence} $\forall(\gamma_1, \gamma_2,\ldots)\in{\cal E}_B,\exists i\geq 1, \gamma_i\in A$
\item \textbf{closure} $\forall(\gamma_1, \gamma_2,\ldots)\in{\cal E}, \gamma_1\in A \Rightarrow \forall i, \gamma_i \in A$
\end{itemize}
\end{definition}

\begin{definition}[Probabilistic attractor] \label{probattractor}
Let $B\subset \C$ and $A\subset B$, $A$ is a probabilistic attractor of $B$ if and only if: 
\begin{itemize}
\item \textbf{convergence} $\forall(\gamma_1, \gamma_2,\ldots)\in{\cal E}_B, \exists i \geq 1, \gamma_i \in A$ whp.
\item \textbf{closure} $\forall(\gamma_1, \gamma_2,\ldots)\in{\cal E}, \gamma_1\in A \Rightarrow \forall i, \gamma_i \in A$.
\end{itemize}
\end{definition}

This means that starting from any configuration in $B$, the system eventually reaches a configuration in $A$ \emph{whp}: an execution can be built in which the system never reaches $A$, but such an execution requires a sequence of decisions that are less and less likely as time goes by. For instance, it can be imagined that one never wins at head or tail, but the longer one plays, the less probable it is. Once the system has reached a configuration in $A$, it (deterministically) remains in $A$.

\begin{definition}[Specification]
A specification is a predicate on a computation.
\end{definition}

\begin{definition}[Self-stabilization]
A  system is self-stabilizing if and only if there exists a non-empty set $\LC \subset \C$ such that
\begin{itemize}
\item $\LC$ is an attractor for $\C$.
\item Every $e$ in ${\LC}$ meets the problem specification.
\end{itemize}
\end{definition}
	
\begin{definition}[Probabilistic self-stabilization]
A  system is probabilistically self-stabilizing if and only if there exists a non-empty set $\LC \subset \C$ such that
\begin{itemize}
\item $\LC$ is a probabilistic attractor for $\C$.
\item Every $e$ in ${\LC}$ meets the problem specification.
\end{itemize}
\end{definition}

Thus, a probabilistically self-stabilizing algorithm is such that the longer one waits, the less likely the algorithm does not meet the specification. This probability can be bounded by a quantity that tends to 0.


\subsection{Random walks properties}

A random walk is a sequence of vertices visited by a token that starts at $i$ and visits other vertices according to the following transition rule: if the token is owned by $i$ at time $t$ then at time $t+1,$ it will be owned by one of its neighbors, this neighbor being chosen uniformly at random among all of them  \cite{Lova93,AKLL+79}. 

\begin{algorithm}[H]
\begin{algorithmic}
\STATE {\bf R0: Upon reception of a message ($Token$)}
\STATE Choose $i$ uniformly at random in $N_i$
\STATE Send $Token$ to $i$
\end{algorithmic}
\caption{Random walk circulation algorithm on  site $i$}
\label{RWAlg}
\end{algorithm}

To compute the complexity of a random walk based distributed algorithm, we use three main quantities:\begin{itemize}
\item The \emph{hitting time} is the average time to reach a node $j$ starting from a node $i$, and is denoted by $h_{ij}$. It is defined as the conditional expectation of the random number of transitions before entering $j$ for the first time knowing that the token starts from $i$.  It has been proven in \cite{Lova93}  that $h_{ij}$ is bounded by $\frac4{27}n^3$. In \cite{IKOY02, NOSY10}, authors provide a local mechanism to reduce this value to $n^2$.
\item The \emph{cover time} is the expected time for a random walk starting at $i$ to visit all the nodes of the system and is denoted by $C_i$. So, the cover time of a graph is $C=\max\{C_i/i\in V\}$ and it was proven in \cite{Feig95,Feig95a} that, depending on the topology of $G$, $O (n\ln n) \leq C \leq  O(n^3)$.
\item Finally, the \emph{meeting time} is the expected time for several random walks to meet on an arbitrary node and is denoted by $M$. The meeting time is bounded by $O(n^3)$ \cite{TeWi91}.
\end{itemize}

\subsection{Problem Specifications}

The specification we are willing to meet is the following, to ensure a consistent token circulation:\begin{itemize}
\item at each step, exactly one \emph{Token} message circulates in the system;
\item any node will receive the token message infinitely often \emph{whp}.
\end{itemize}

Since we suppose the treatment of messages is atomic, and take the successive configurations of the system when the considered node has finished with its local treatment, we need not consider the case when the token is being treated.

We note $Token_\gamma$ the set of token messages at configuration $\gamma$, and $Token_\gamma(i)$ the set of token messages heading to a node $i$

\begin{definition}[Problem Specification]\label{PS}
We say a computation $\mathcal E=(\gamma_1, \gamma_2, \ldots)$ satisfies specification $ProbTokCirc$ of Probabilistic Token Circulation if:
\begin{itemize}
\item $\forall k, |Token_{\gamma_k}|=1$ (there exists exactly one token in the system);
\item $\forall k, \forall i, \exists l>k, |Token_{\gamma_{l}}(i)|=1$ \emph{whp} (any node will receive the token infinitely often).
\end{itemize}

\end{definition}


Our contribution is to design a solution that eventually satisfy these specifications in a dynamic and faulty environment.

\section{Dealing with topology changes} \label{TC}

We show in this section that a random walk on a dynamic graph has the same properties that a random walk on a static graph. A random walk is well adapted to dynamically evolving graphs. Indeed, its traversal is based only on local information, and it has not to be redesigned after a topological change. We prove in this section that if the node mobility is independent from the token moves (in particular, if no daemon picking the random moves of the token is behaving as an adversary), desirable properties hold:\begin{itemize}
\item any node is visited in finite time;
\item we can compute the average time it takes to hit a given node, or to visit all nodes.
\end{itemize}

Consider a dynamic graph on a static set of nodes, with dynamic edges $G_t=(V, E_t)$, with $t$ a continuous time index. We model the disconnection of a node by all its link being removed. We suppose in the sequel that:\begin{itemize}
\item the evolution of the graph is an homogeneous Markov process (\emph{ie} the evolution of the system topology only depends on its current state);
\item it is independent from the choices of the random walk (this avoids cases with the system behaving as an opponent to the walk).
\end{itemize}

The homogeneity assumption means that the token evolves much faster than the system. Clearly, in most concrete applications, the system evolution is driven by some daily cycle. If the evolution is weak at a time scale of below one minute, and that the hitting time is itself below one minute, then, this assumption is realistic in the following computations.

Thus, if the system is considered at each reception of the token, the evolution of the graph is discretized. In the sequel, we consider the discretization $G_\gamma$, which is a Markov chain by independence of the token movements and of the graph evolution.

Given the graphs $G$ and $G'$, we note $p_{G\rightarrow G'}$ the probability that at a step the dynamic graph is $G$ and at the next step, it is $G'$. A step corresponds to a time unit (cf. Section \ref{MDS}). 
If the graph evolves as a markovian process, then this discretization is a Markov chain.

We note $p_{ij}(G)$ the probability that, in $G$, node $i$ sends the token to $j$. $h_{ij}(G)$ is the average time it takes to the walk, starting on $i$, to reach $j$, knowing that at the beginning of the walk, the system is in the state described by $G$. Finally, the system being described as an homogeneous Markov chain with a non-bipartite finite state space, it has a stationary distribution we note $\pi$. The state space is non-bipartite, since the opposite would mean that edges blink at the exact same pace as the token moves. A stationary distribution means that if we look at the system at a certain time, then with probability $\pi(G)$, its topology is $G$.

Note $\overline{p_{ij}}=\sum_G\pi(G)p_{ij}(G)$ the average probability that the token being on $i$, it is sent to $j$.

\begin{theorem}\ref{hittingdyn}
The hitting time of a random walk on a dynamic graph is such that $h_{ij}=1+\sum_k\overline{p_{ik}}h_{kj}$, and $h_{ii}=0$.
\end{theorem}
\begin{proof}
We state that:
$$h_{ij}(G)=\sum_k p_{ik}(G)\left(1+\sum_{G'} p_{G\rightarrow G'}h_{kj}(G')\right)$$

This means that the token being in $i$ and the system being described by $G$, with probability $p_{ik}(G)$, the token is sent to $k$. When sent to $k$, it takes one step and then, the token has to go from $k$ to $j$, the system topology having evolved to $G'$ (with probability $p_{G\rightarrow G'}$) in the meantime. The hitting time from $i$ is one step to send the token to one neighbor, plus the expectation over the chosen neighbor of the average time it takes to the token to go from it to $j$. This hitting time from $k$ to $j$ is the average hitting time over the possible system states $G'$.

We are interested in the hitting time from node $i$ to node $j$. If we have no information on the system state at the beginning of the process, we take the average hitting time over all possible system states: $h_{ij}=\sum_G \pi(G)h_{ij}(G)$.

\begin{equation*}\begin{split}
h_{ij}&=\sum_G \pi(G)h_{ij}(G)\text{ by definition} \\
&=\sum_G\pi(G)\sum_k p_{ik}(G)\left(1+\sum_{G'}p_{G\rightarrow G'}h_{kj}(G')\right)\text{ according to the previous equation}\\
&=\sum_G\sum_k\pi(G)p_{ik}(G)+\sum_G\pi(G)\sum_k p_{ik}(G)\sum_{G'}p_{G\rightarrow G'}h_{kj}(G')\\
&=\sum_k\sum_G\pi(G)p_{ik}(G)+\sum_k\sum_{G'} \sum_G\pi(G)p_{G\rightarrow G'}p_{ik}(G)h_{kj}(G')\\
&=\sum_k\overline{p_{ik}}+\sum_k\sum_{G'} \sum_G\pi(G)p_{G\rightarrow G'}p_{ik}(G)h_{kj}(G')\text{ by definition of $\overline{p_{ik}}$}\\
&=\sum_k\left(\overline{p_{ik}}+\sum_{G'} \sum_G\pi(G)p_{G\rightarrow G'}p_{ik}(G)h_{kj}(G')\right)\\
&=\sum_k\left(\overline{p_{ik}}+\sum_{G'}h_{kj}(G')\sum_G\pi(G)p_{G\rightarrow G'}p_{ik}(G)\right)\\
&=\sum_k\left(\overline{p_{ik}}+\sum_{G'}h_{kj}(G')\left(\sum_G\pi(G)p_{G\rightarrow G'}\right)\left(\sum_G\pi(G)p_{ik}(G)\right)\right)
\end{split}
\end{equation*}

The last equality comes from the independence of the system evolution and the token moves. Indeed, $\sum_G\pi(G)p_{G\rightarrow G'}p_{ik}(G)$ is the expectation over the system states of the probability that the system evolves to a given state $G'$ times the probability that the token moves to $k$. Since these quantity are independent, the expectations of their product is the product of their expectations, $\left(\sum_G\pi(G)p_{G\rightarrow G'}\right)\left(\sum_G\pi(G)p_{ik}(G)\right)$.

Thus, since $\sum_G\pi(G)p_{G\rightarrow G'}=\pi(G')$ (by definition of a stationary distribution) and $\sum_G\pi(G)p_{ik}(G)=\overline{p_{ik}}$ (by definition):
\begin{equation*}\begin{split}
h_{ij}&=\sum_k\left(\overline{p_{ik}}+\sum_{G'}h_{kj}(G')\left(\sum_G\pi(G)p_{G\rightarrow G'}\right)\left(\sum_G\pi(G)p_{ik}(G)\right)\right)\\
&=\sum_k\left(\overline{p_{ik}}+\sum_{G'}h_{kj}(G')\pi(G')\overline{p_{ik}}\right)\text{ by definition of $\overline{p_{ik}}$}\\
&=\sum_k\overline{p_{ik}}\left(1+\sum_{G'}h_{kj}(G')\pi(G')\right)\\
&=\sum_k\overline{p_{ik}}(1+h_{kj})\\
\end{split}
\end{equation*}

Now, $\overline{p_{ij}}$ is a transition probability: 
\begin{equation*}\begin{split}
\sum_j\overline{p_{ij}}&=\sum_j\sum_G\pi(G)p_{ij}(G)\\
&=\sum_G\sum_j\pi(G)p_{ij}(G)\\
&=\sum_G\pi(G)\sum_jp_{ij}(G)\\
&=\sum_G\pi(G)=1
\end{split}\end{equation*}

Finally, $h_{ij}=1+\sum_k\overline{p_{ik}}h_{kj}$, and $h_{ii}=0$, which is the very equation followed by the hitting time of a random walk on a weighted graph $\overline G=(E, E\times E, \omega)$, with $\omega(i, j)=\overline{p_{ij}}$.
\end{proof}

\begin{corollary}\label{RWProp}
Random walks on dynamic graphs verify the hitting, cover, and (if the graph is not bipartite) meeting properties.
\end{corollary}
\begin{proof}
The computation of the hitting time at theorem \label{hittingdyn} implies that it is finite. Thus the hitting property is verified for any node. The cover property follows from the hitting property.
\end{proof}

\begin{corollary}
Algorithms in \cite{BuSo07} that compute hitting and cover times apply.
\end{corollary}

However, $\omega(i, j)$ and $\omega(j, i)$ can be different, which would make the graph directed. Classical bounds on hitting times and cover times may not apply.

\section{A self-stabilizing token maintenance mechanism}\label{SSTMM}

\subsection{Principles}
In this section, we focus on the token maintenance mechanism. For the sake of clarity, we assume in this section that the number of nodes in the system is exactly $n$. We will discuss in the next section how to relax this assumption. 

We consider a token that circulates through  the system using a random walk scheme, thus by corollary \ref{RWProp}, all nodes are visited infinitely often (satisfying the first part of the specification, cf. Definition \ref{PS}). The system can be erroneously initiated: configuration with no token, or with several tokens can occur.  To solve the absence of token, we introduce a content in the token. As the designed solution is self-stabilizing, we have to deal with arbitrary initiated token content.


\subsubsection{Dealing with the absence of token}
The lost token situation is solved by a decentralized timeout procedure:  each processor indistinctly has the possibility of producing a new token. 
Each node maintains a timer. Each timer is set at a value $T_m$ time units. (The way to tune the values of $T_m$ will be discussed in Section \ref{Tt}). 
Each node measures the time since the last token visit. If this time is greater than $T_m$, then a new token is created. 
No upper bound is available on the time the token returns to a node $i$, which makes it impossible to use solutions like the one in \cite{Varg00}, consisting in setting a timeout on each node, at the expiration of which, if no token has been received, a new one is created. The following impossibility result proves this.




\begin{proposition}[Impossibility result]
Whatever the timers values of each node in the system,  the closure property is not satisfied \emph{whp}. 
\end{proposition}

\begin{proof}
Let $i$ be a node in $G$, with (at least) two neighbors $j$ and $k$ (in a connected graph with more than two nodes, such a node exists). Consider a legitimate configuration with the token on $i$ and a timer $T$ on $k$. Note $d=\max\{\deg(j), \deg(i)\}$. Considering the case when the token is on $i$, with probability greater than $\frac1d$, it goes to $j$. Then, with probability greater than $\frac1d$, it goes to $i$. Then, with probability greater than $\frac1{d^T}>0$, the token does not hit $k$ for $T$ steps, leading to its timeout being triggered, and an unnecessary token creation. Thus, \emph{whp}, the system spontaneously leaves a legitimate configuration.
\end{proof}

To avoid unnecessary token creation, we propose a solution with the following mechanisms:
\begin{enumerate}
\item A local mechanism for monitoring the last visit time of the token to a node $i$. 
\item A mechanism for detecting that some timers are about to expire. This mechanism is  maintained by the node which is the current token holder.
\item A mechanism maintaining a spanning tree, which is rooted at the current token holder.
\item A distributed mechanism that propagates messages on this tree in order  to reset the timers.
\end{enumerate}


The first two items correspond to the decentralized timeout procedure. 
The last two items correspond to the reloading wave. 
\paragraph{(3) Reloading wave definition}
A reloading wave is defined  regarding a token identity.  

When a node $i$ receives reloading wave message, $i$ is notified that a token is still circulating in the system, and it resets its timer. Thus, the reloading wave prevents node $i$ from creating a copy of token $t$. 

The reloading wave is broadcast through an adaptive (spanning) tree. There is no additional protocol, since we use token $t$ content. The  token collects and stores the identities of each node during its random walk traversal. This content is based on the history of the token's moves. Such a token is called a \emph{circulating word}. Since the token is circulating continuously through the network, the induced tree is perpetually updated taking into account the possible network topology changes.

A token $t$ contains the following data structures:
\begin{itemize}
\item A counter, $t.hop$ that represents the number of edges visited during the traversal.
\item An array, $t.table$. Each time the token moves from a node $j$ to a node $i$,  the token sets $t.table[j] = i$, and $t.table[i]=i$.
\end{itemize}

Each time a node $i$ receives the token $t$, a tree rooted on $i$ can be locally computed by $i$ using the topological information stored in the token.
The tree induced by $t.table$ is $(V,E_T)$ where $E_T = \{ (k, t.table(k)), k\in V \mbox{and } k\not = t.table(k) \}$.

\begin{example}
From the following sequence of the token's moves $<1,3,5,4,3>$ the token is at node 3 and $t.table$ is ($\bot$ represents the value \emph{``undefined''}):
$
\begin{tabular}{|c|c|c|c|c|}
\hline 
1&2&3&4&5 \\ \hline
3 &$\bot$&3&3&4 \\ \hline
\end{tabular}
$

The forest induced by $t.table$ is $(V,E_T)$ where $V=\{ 1,3,4,5\} \cup \{2\}$ and 
$E_T= \{ (1,3), (4,3), (5,4) \}$. 

If the next token moves are $<2,1,2,3,1>$ the token table is updated to  
 $
\begin{tabular}{|c|c|c|c|c|}
\hline 
1&2&3&4&5 \\ \hline
1 &3&1&3&4 \\ \hline
\end{tabular}
$
and the tree induced by $t.table$ is now $\{V=\{ 1,2,3,4,5\} $ and 
$E_T= \{ (3,1), (2,3), (4,3), (5,4) \} \}$. 

\end{example}

\paragraph{(4) The reloading wave mechanism} The reloading wave is broadcast under the following conditions:
the token maintains the counter $t.hop$ which is incremented at each hop. The counter is set to 0 at token creation.  This value is compared to the timeout value $ T_m$ minus the time to achieve a wave propagation. 
When this counter value is superior or equal to the latter value, the node that holds the token launches the wave and the token counter $t.hop$ is reset to 0.
When a node receives the wave from  the token, it reloads its timer  to $T_m$.

\subsubsection{Configurations with multiple tokens}

To design a self-stabilizing solution, starting from any initial configuration, the system must converge to a correct behavior: exactly one random walk based token circulates through the system. The previous section deals with the way to produce at least a token when a communication deadlock occurs. 
Faulty configurations with several tokens are possible (due to duplication for instance). 

Various articles \cite{IsJa90,TeWi91} have dealt with the multiple token situation in case of a random walk scheme. The authors propose to use the meeting property (cf. Corollary \ref{RWProp}) of random walks to reduce the number of token to 1: each time a node receives several tokens, it discards all of them but one. Thus, in finite time, a single token remains in the network.

Two strategies are possibles: 
\begin{itemize}
\item remove all tokens but one;
\item merge the content of all tokens in a new one.
\end{itemize} 

We propose to merge all the topological information before discarding any token. This strategy entails more computation, but accelerates the construction of a spanning tree inside a token. Once all the sub-trees contained in the different tokens have been merged (cf. Procedure \ref{merge}), the resulting sub-tree is stored in the remaining token (the one that is not discarded, cf. Rule R1.b Algorithm \ref{Main}).  


\begin{procedure}[H]
\begin{algorithmic}
\FOR{$k = 0$ to $\mathcal N$}
 \IF{$(t1.table[k] = \bot)\wedge (t2.table[k] \neq \bot)$}
  \STATE $t1.table[k]\rcp t2.table[k]$
 \ENDIF
\ENDFOR
\STATE $t1.hop\rcp \max(t1.hop, t2.hop)$
\end{algorithmic}
\caption{Procedure: merge\_tokens(t1: token, t2: token) on node $i$}
\label{merge}
\end{procedure}

\begin{figure}[H]
\begin{center}
\begin{tabular}{cc}
\includegraphics[scale=0.3]{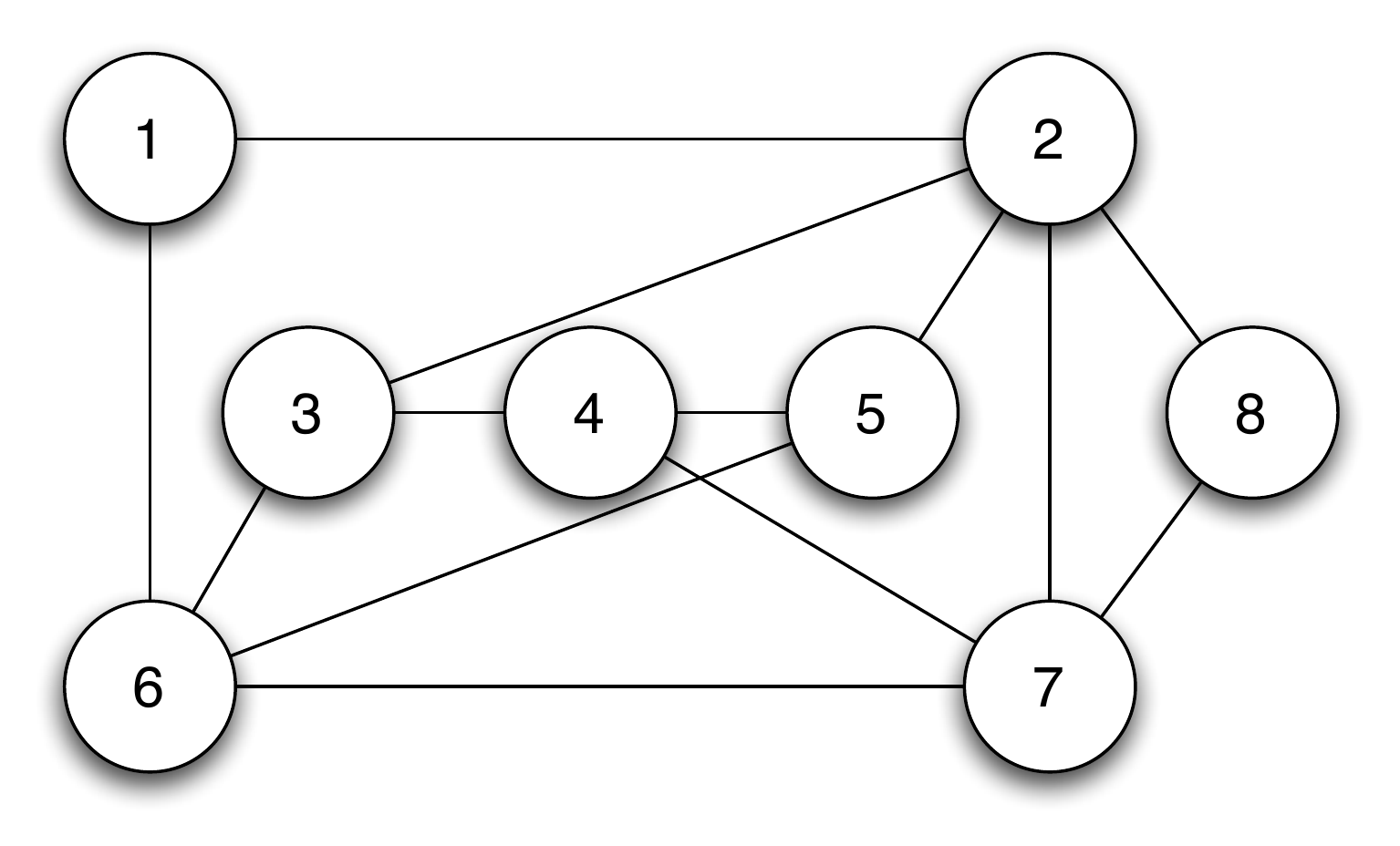}&
\includegraphics[scale=0.3]{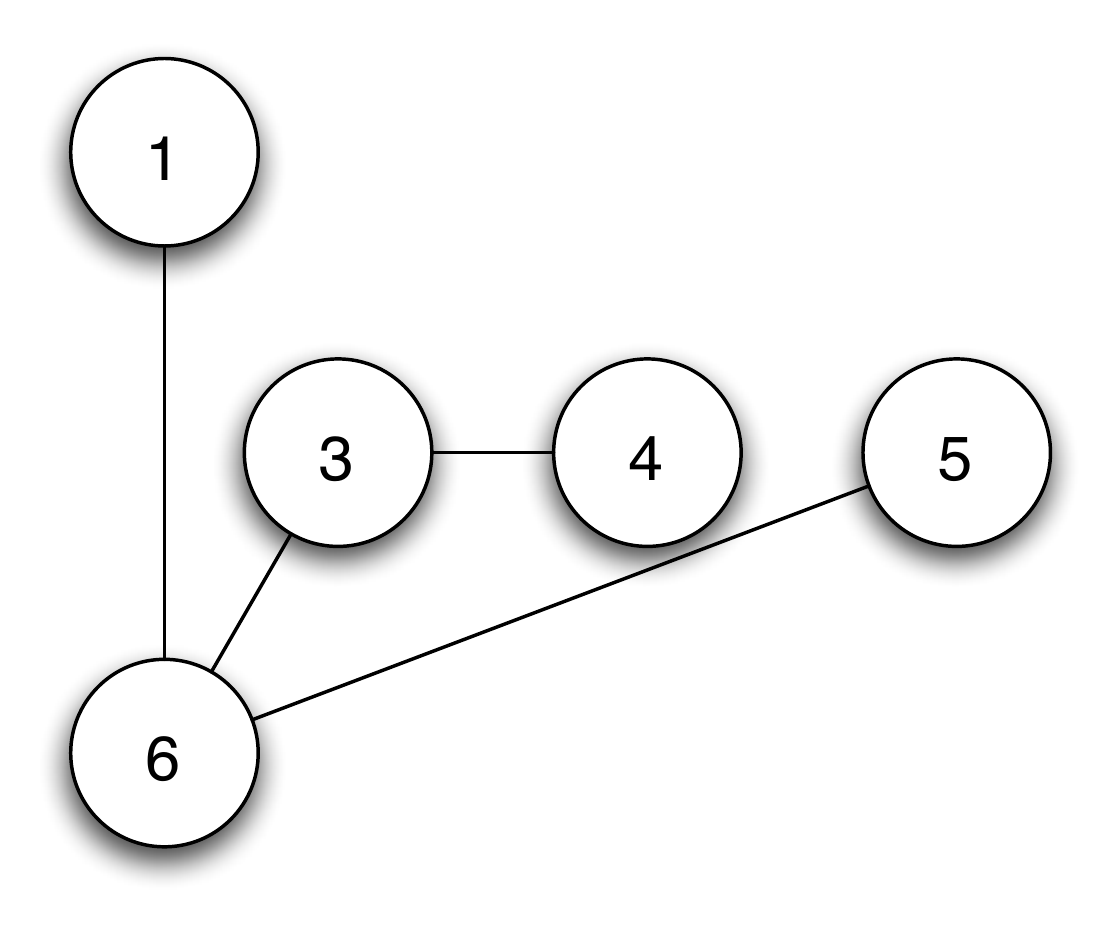}\\
Network&token $t1$\\
\includegraphics[scale=0.3]{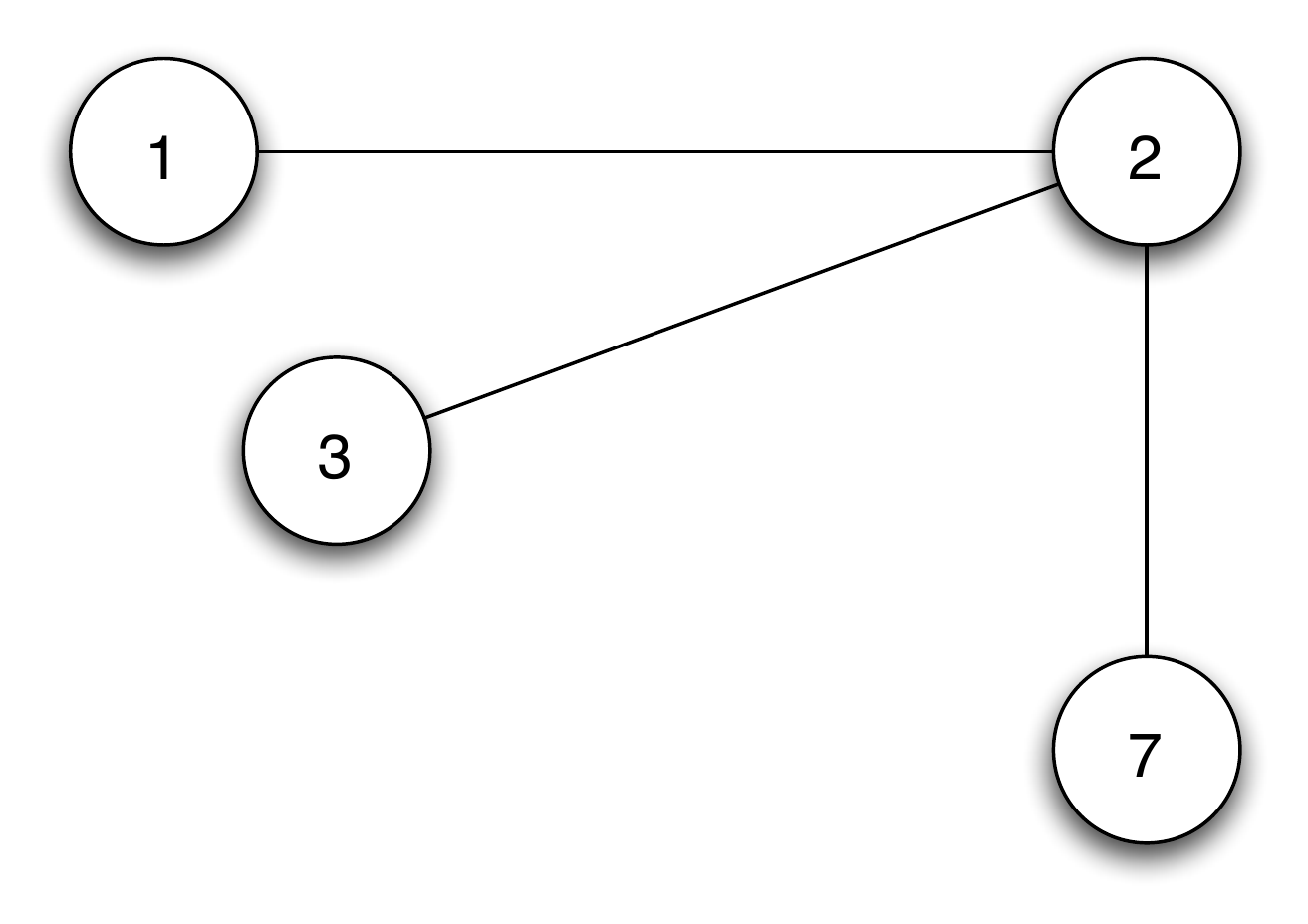}&
\includegraphics[scale=0.3]{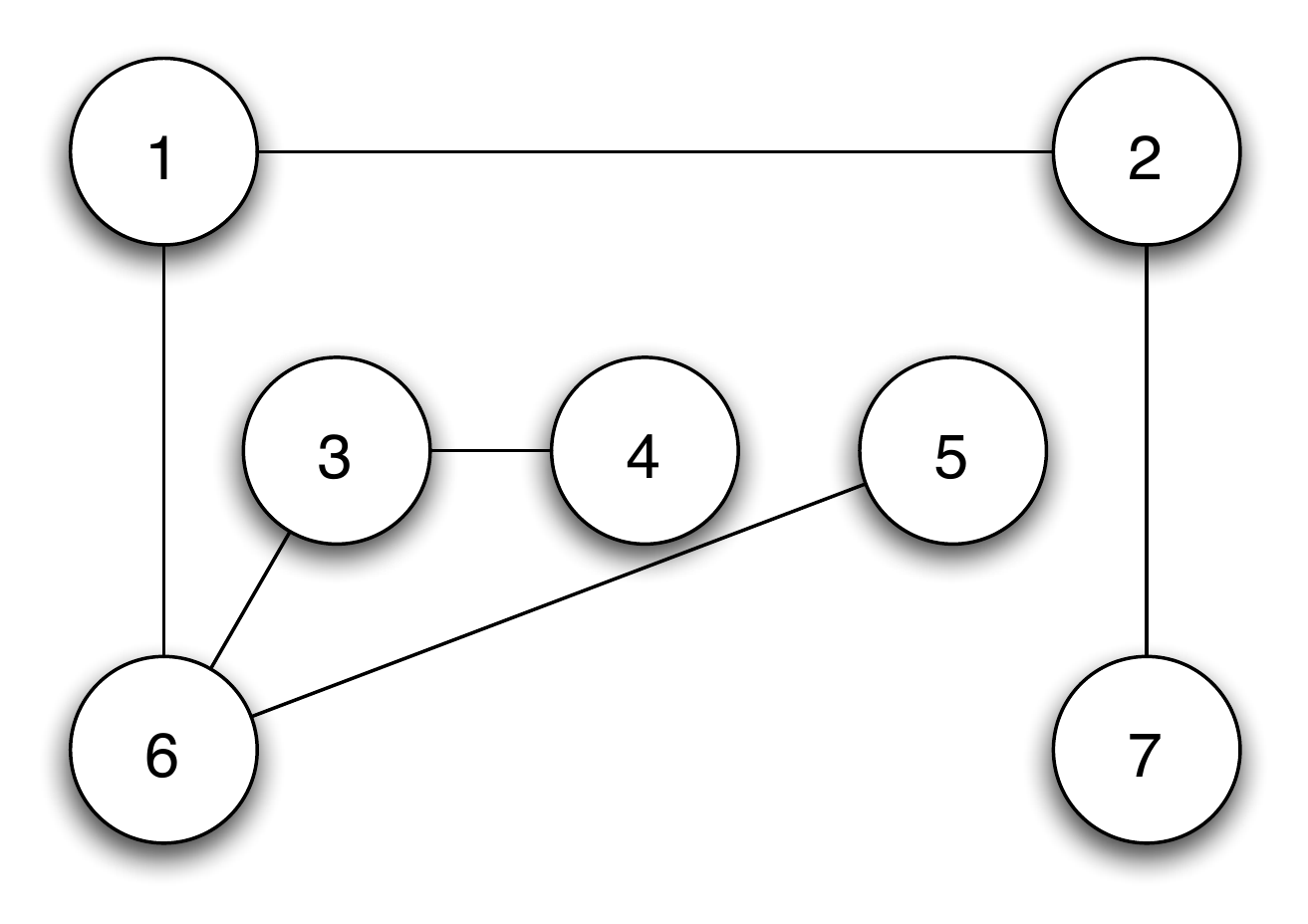}\\
token $t2$& Resulting token\\
\end{tabular}
\caption{Example of two merged tokens}\label{}
\end{center}
\end{figure} 

\color{black}

\subsubsection{Configuration with arbitrary token content}
Transient failures can produce erroneous token content: a node $j$ may be registered as father of node $i$ in the reloading wave tree while they are not neighbors. Then, when $i$ is hit by the token, its father is set to itself, and the error is corrected. 


\subsection{The algorithm}

The algorithm is written according 4 events on a node:
\begin{itemize}
\item Node $i$ receives one or several tokens ({\bf R1}). Each token is updated and its consistency is checked ({\bf R1.a}). In case of multiple tokens, they are merged into one ({\bf R1.b}). If the condition to launch the reloading wave is satisfied, the node begins the propagation of the reloading wave ({\bf R1.c}).
Finally the token is forwarded to a neighbor chosen at random ({\bf R1.d}) and the node resets its timer ({\bf R1.e}).
\item The timer of node $i$ expires ({\bf R2}). A new empty token is created ({\bf R2.a}) and forwarded to a neighbor chosen at random ({\bf R2.b}). The node resets its timer ({\bf R2.c}).
\item Node $i$ receives a reloading wave message ({\bf R3}). The node continues the reloading wave propagation ({\bf R3.a}) and resets its timer ({\bf R3.b}).
\item Node $i$'s clock ticks ({\bf R4}). The node decrease its timer. 
\end{itemize}

This algorithm has four rules \textbf{R1} to \textbf{R4}, that are split into sections. All sections in a rule are executed in sequence, and the rule is executed atomically.

\begin{algorithm}[H]
\begin{algorithmic}
\STATE {\bf R1: Upon reception of a set $T$ of Token messages}
\STATE {\it\underline{a: Tokens update} }
\FORALL{$t\in T$}
\STATE $t.table[i]\rcp i$
\STATE $t.table[t.emitter]\rcp i$
\STATE $t.hop \rcp t.hop + 1$
\ENDFOR
\STATE {\it \underline{b: Tokens merge }}
\STATE choose $t$ in $T$
\STATE $T\rcp T\backslash \{t\}$
\STATE $t'\rcp t$
\WHILE{$T\neq\emptyset$}
\STATE choose $t_2$ in $T$
\STATE merge\_tokens($t'$, $t_2$)
\STATE $T\rcp T\backslash \{t_2\}$
\ENDWHILE
\STATE {\it \underline{c: Possible Reloading Wave propagation}}
\IF{$t'.hop \geq T_m -  (n+1)$} 
 \FORALL{$j$ such that $t'.table[j]=i\wedge j\in N_i$}
  \STATE send $reload, t'.table$ to $j$
 \ENDFOR
 \STATE $t'.hop\rcp 0$
\ENDIF
\STATE {\it \underline{d: Token circulation }}
\STATE Send $t'$ to $j$ chosen randomly in $N_i$
\STATE {\it \underline{e: Node update}}
\STATE $timer \rcp T_m$
\vspace{0.5cm}
\STATE {\bf R2: Upon a release of timer}
\STATE {\it \underline{a: Token Creation}}
\FOR{$j=0$ to $\cal N$}
 \STATE $t'.table[j]\rcp \bot$
\ENDFOR
\STATE $t'.table[i]\rcp  i$
\STATE $t'.hop \rcp 0$
\STATE {\it \underline{b: Token Circulation}}
\STATE Send $t'$ to $j$ chosen randomly in $N_i$
\STATE {\it \underline{c: Node update}}
\STATE $timer\rcp T_m$
\vspace{0.5cm}
\STATE {\bf R3: Upon a reception of message ($reload,table$)}
\STATE {\it \underline{a: Reloading Wave propagation }}
\STATE $table[i]\rcp \bot$ \COMMENT{to ensure that the reloading wave terminates}
\FORALL{$j$ such that $table[j]=i\wedge j\in N_i$}
 \STATE send $reload,table$ to $j$
\ENDFOR
\STATE {\it \underline{b: Node update }}
\STATE $timer \rcp T_m$
\vspace{0.5cm}
\STATE {\bf R4: Upon a clock tick}
\STATE $ timer \rcp timer -1$
\end{algorithmic}
\caption{Algorithm on  site $i$}
\label{Main}
\end{algorithm}

\subsection{Proofs}
We present in this section the correctness proofs of the algorithms. We show that our algorithm is self-stabilizing  and achieve a token circulation: the execution of our algorithm starting in an arbitrary configuration will reach a \emph{legitimate configuration} (the set $\LC$ of configurations). 

\subsubsection{Preliminaries}

A configuration $\gamma$ is characterized by:
\begin{itemize}
	\item the graph $G_\gamma=(V_\gamma, E_\gamma)$; in this section, no topological change is assumed, so that this graph is constant $G_\gamma=G=(V, E)$;
	\item the value of variables:\begin{itemize}
		\item the value of all timers $timer_i(\gamma)$ for all $i\in V_\gamma$\end{itemize}
	\item the multi-set of messages, composed of:
	\begin{itemize}
		\item $Token_{\gamma}$ the multiset of token messages, in $E\times V^V\times [0; T_m]$: $t=((i, j), table_t, hop_t)\in Token_\gamma$ means that there is (at least) one token $t$ sent from $i$ and pending reception by $j$ with table $table_t$ and hop counter $hop_t$; we note $Token_{\gamma}(i)=\{((j, i), table_t, hop_t)\in Token_{\gamma}\}$ the set of all token messages pending reception by $i$; for $t=((j, i), table, hop)\in Token_{\gamma}$, we note $t.emitter=j$, $t.recipient=i$, $t.table=table$ and $t.hop=hop$;
		\item $Wave_\gamma$ the multiset of reloading wave messages, in $E\times V^V$: $w=((i, j), table_w)\in Wave_\gamma$ means that there is (at least) one reloading wave message $w$ sent from $i$ and pending reception by $j$ with table $table_w$.
	\end{itemize}
\end{itemize}


We consider that the execution of an algorithm is atomic.

First, we define what we call a token, and then we prove that the reloading wave has the intended effect: no token can be created by a node that has already received a token. Finally, we prove that the algorithm provides a self-stabilizing traversal scheme.



Consider two configurations $\gamma\vdash\gamma'$ (the execution being supposed atomic, such a step involves that a message has been received and treated to reach $\gamma'$ from $\gamma$). If $\gamma'$ is the result of the application of \textbf{Rk} ($1\leq k\leq4$) by node $i$ we note $\gamma\vdash^{\mathrm{\mathbf{Rk}}(i)}\gamma'$. The execution of all algorithms being supposed atomic, if $\gamma\vdash\gamma'$, we have the following possibilities:

\begin{enumerate}
\item \label{jetons} $\gamma\vdash^{\mathrm{\mathbf{R1}}(i)}\gamma'$: then, $\gamma$ is such that $\exists ((k, i), t)\in T\subset Token_\gamma$, and $\gamma'$ is obtained from $\gamma$ by:
	\begin{enumerate}
	\item\label{countjeton} $t'.hop=\max\{t.hop/t\in T\}+1\mod (T_m-(n+1))$;
	\item\label{tablejeton} $\forall k\neq i, j, (\exists t\in T, t'.table[k]=t.table[k]\neq\bot)\vee(\forall t\in T, t.table[k]=\bot)$; $t'.table[i]=t'.table[j]=i$;
	\item\label{wavejeton} if $t'.hop=0$, $Wave_{\gamma'}=Wave_\gamma\cup\{((i, j), t'.table)/t'.table[j]=i\wedge j\in N_i\}$;
	\item\label{envoijeton} $Token_{\gamma'}=Token_{\gamma}\backslash T\cup\{((i, j), t')\}$, with $j\in N_i$;
	\item\label{timerjeton} $timer_i^{(\gamma')}=T_m$;
	\end{enumerate}
\item $\gamma\vdash^{\mathrm{\mathbf{R2}}(i)}\gamma'$: then, $\gamma$ is such that $timer_i^{(\gamma)}=0$, and $\gamma'$ is obtained from $\gamma$ by:	
	\begin{enumerate}
	\item\label{tablevide1} $\forall j\neq i, t'.table[j]=\bot$; $t'.table[i]=i$;
	\item\label{envoijetonvide} $Token_{\gamma'}=Token_{\gamma}\cup\{((i, j), t')\}$ with $j\in N(i)$;
	\item\label{timercrea} $timer_i^{(\gamma')}=T_m$;
	\end{enumerate}
\item $\gamma\vdash^{\mathrm{\mathbf{R3}}(i)}\gamma'$: then, $\gamma$ is such that there is $((k, i), w)\in Wave_\gamma$, and $\gamma'$ is obtained from $\gamma$ by:
	\begin{enumerate}
	\item\label{vague} $Wave_{\gamma'}=Wave_{\gamma}\backslash\{((k, i), w)\}\cup\{((i, j), w')/w.table[j]=i\wedge j\in N_i, \forall k, w'.table[k]=w.table[k], w'.table[i]=\bot\}$;
	\item\label{timervague} $timer_i^{(\gamma')}=T_m$;
	\end{enumerate}
\item\label{clocktick} $\gamma\vdash^{\mathrm{\mathbf{R4}}(i)}\gamma'$: then, $\gamma$ is such that there is $timer_i^{(\gamma)}>0$, and $\gamma'$ is obtained from $\gamma$ by:
	\begin{enumerate}
	\item $timer_i^{(\gamma')}=timer_i^{(\gamma)}-1$.
	\end{enumerate}
\end{enumerate}

In item \ref{jetons}, $T$ represents the set of tokens that are received by $i$. $T$ contains at least one token, but may contain several of them, in which case they are merged into one token noted $t'$ in the sequel. \ref{countjeton} is the update of the hop counter: the hop counter is decreased by one, and if it reaches 0, a wave is propagated (\ref{wavejeton}) and the hop counter reset (hence the $\mod T_m-(n+1)$). Node $i$ resets its timer (\ref{timerjeton}). \ref{tablejeton} is the computation of the new table: $i$ is the root, and the father of the sender, the remaining of the tree is obtained by picking for each node of the tree its father in one of the received trees.

At the timer expiration on node $i$, it sends a newly created token. \ref{tablevide1} is the creation of a tree consisting of the single node $i$. At \ref{timercrea}, the timeout is reset. \ref{envoijetonvide} states that, at some edge neighboring $i$, the new token is added.

In \ref{vague}, node $i$ receives a $Wave$ message $w$ and sends $Wave$ messages $w$ to all its children as indicated in $w.table$. It resets its timer (\ref{timervague}).

At each clock tick, node $i$ decrements its timer (\ref{clocktick}).

Between two successive applications of \textbf{R4} by a given node, all nodes that can apply \textbf{R3}, \textbf{R1} and \textbf{R2} apply them. In rules \textbf{R1} and \textbf{R2}, node $j$ is chosen at random.

\begin{definition}[Token and state of a token]
From \textbf{R1}, we say that any token in $T$ has become $t'$. For $t$ in $T$, we will note $t^{(\gamma)}\rightarrow t^{(\gamma')}$.
\end{definition}

\begin{definition}
A node $i$ is said to receive a token at step $\gamma\rightarrow \gamma'$ if there exists a token $t$ in an edge to $i$ at configuration $\gamma$, with $t\rightarrow t'$, and $t'$ is in an edge from $i$.
\end{definition}

All tokens follow a random walk. In particular, the hitting and cover properties are verified, so that, for any node $i$ and any token $t$ in a configuration $\gamma$, there exists a configuration $\gamma'$ in $\mathcal E_\gamma$ such that $t^{(\gamma')}$ is in an edge coming from $i$.

\subsubsection{All tokens are eventually correct}

This step needs no synchronism. Basically, the only needed property is that the random walk covers the system, \emph{ie} that the random numbers generators are independent.

\begin{definition}We say that a token $t$ is \emph{correct} and we note $correct(t)$ if $$\forall k\in V, t.table[k]\neq\bot\Rightarrow(k, t.table[k])\in E$$\end{definition}

\begin{lemma}\label{1}
$\mathcal{A}_1=\{\gamma \in \mathcal{C}/ \forall t\in Token_\gamma, correct(t)\}$ is an attractor of $\mathcal{C}$
\end{lemma}

\begin{proof}
Note $inc(\gamma)=\{(t, i)\in Token_\gamma\times V/(i\neq t.emmiter\wedge t.table[i]\neq\bot\wedge(i, t.table[i])\notin E)\vee(i=t.emitter\wedge t.table[i]\neq i)\}$. A token is correct if and only if it does not appear in this set. We will show that eventually, $inc(\gamma)=\emptyset$. First, we show that it is non-increasing, and then that, if it is greater than 0, then it eventually decreases.

\textbf{R3} and \textbf{R4} do not affect $Token_\gamma$, and in consequence $inc(\gamma)$.

\textbf{R2} creates a new empty token $t$. This token is correct: for all $j\neq t.emitter$, $t.table[j]=\bot$, so that $(t, j)\notin inc(\gamma)$. Since other tokens are left unchanged, $|inc(\gamma)|$ does not increase.

Consider the case when \textbf{R1} is applied by node $i$ to a set of tokens $T$. Then, $Token_{\gamma'}=Token_{\gamma}\backslash T\cup\{t'\}$ with $t'.table[i]=i$, $\forall j\neq i, \exists t\in T, t'.table[j]=t.table[j]$. Thus, each inconsistency in $t'$ comes from an inconsistency in a token in $T$: if $(t', j)\in inc(\gamma)$, there exists (at least) a token $t$ in $T$ such that $(t, j)\in inc(\gamma)$. Thus, $inc(\gamma')\subset inc(\gamma)$. Now, if $t\in T$ is such that $(t, i)\in inc(\gamma)$, $t'.table[i]=i$, and $(t', i)\notin inc(\gamma')$, so that $inc(\gamma')\subsetneq inc(\gamma)$ (note that by merging several tokens, some other inconsistencies may be corrected).

Thus, $inc(\gamma)$ does not increase. Now, consider a configuration $\gamma$ such that $inc(\gamma)\neq\emptyset$. Then, there exists $(t, i)\in inc(\gamma)$. The hitting property entails that $t$ will eventually hit $i$ at configuration $\gamma'$, and then $inc(\gamma')\subset inc(\gamma)\backslash\{(t, i)\}$.

Thus, if $inc(\gamma)\neq\emptyset$, it eventually decreases. Eventually, it reaches $\emptyset$, and then, all tokens are correct.\end{proof}

\subsubsection{There is eventually a correct token (at least) in the system}

This step requires that if a single rule is enabled, it is eventually triggered.

\begin{lemma}\label{2}
$\mathcal{A}_2=\{\gamma \in \mathcal{C}, |Token_\gamma| \geq 1\}$ is an attractor of $\mathcal{C}$
\end{lemma}

\begin{proof}
First, we show that if there is a token in the system, it cannot disappear. Consider a configuration $\gamma$ such that $Token_\gamma\neq\emptyset$ and $\gamma\rightarrow\gamma'$. If $\gamma\rightarrow^\mathrm{\mathbf{R3}(i)}\gamma'$ or $\gamma\rightarrow^\mathrm{\mathbf{R4}(i)}\gamma'$, then $Token_{\gamma'}=Token_\gamma\neq\emptyset$. If $\gamma\rightarrow^\mathrm{\mathbf{R2}(i)}\gamma'$, $Token_{\gamma'}\supset Token_{\gamma}\neq \emptyset$. Last, if $\gamma\rightarrow^\mathrm{\mathbf{R1}(i)}\gamma'$, $Token_{\gamma'}$ contains the token put at \textbf{R1}.d, and is not empty.

Suppose $Token_{\gamma}=\emptyset$. First, we show that $Wave_{\gamma'}$ is eventually empty. Since $Token_\gamma=\emptyset$, if \textbf{R2} is triggered, a token is created and $Token_\gamma$ is no longer empty. Aside \textbf{R2}, the only rules that can be triggered are \textbf{R3} and \textbf{R4}. \textbf{R4} does not modify $Wave$. Consider a message $((i, j), w)\in Wave_\gamma$ received at step $\gamma\rightarrow\gamma'$. $Wave_{\gamma'}=Wave_{\gamma}\backslash\{((k, i)w)\}\cup\{((i, j), w')/w.table[j]=i\wedge j\in N_i, \forall k, w'.table[k]=w.table[k], w'.table[i]=\bot\}$. Thus, since a wave message $v$ is sent to $i$ only by $v.table[i]$, $i$ cannot receive any more message triggered by $w$. Thus, all sites can receive at most one wave message for each wave message present in $Wave_\gamma$. Thus, eventually, $Wave_\gamma=\emptyset$.

Now, the only rules that apply are \textbf{R2} and \textbf{R4}. The continuing application of \textbf{R4} leads to a timeout (or even all of them, leaving \textbf{R2} the only activated rule) to reach 0, so that \textbf{R2} is triggered, and a token created.\end{proof}

\begin{corollary}
$\mathcal A_1\cap \mathcal A_2$ is an attractor of $\mathcal C$.
\end{corollary}

\subsubsection{No visited node can create a token --- Reloading wave and synchronicity}

To verify this property, we need synchronicity assumptions: all nodes timers must be decremented at most once in the time it takes to a token to be received, treated, and sent again. We also need the cover property to be true, so independent random numbers generators on the nodes.

We consider an arbitrary configuration $\gamma_0\in \mathcal A_1\cap \mathcal A_2$. All the following properties are about $\mathcal E_{\gamma_0}=(\gamma_0, \gamma_1, \ldots)$. We consider a configuration $\gamma_i$ in $\mathcal E_{\gamma_0}$, and a token $t$ in $\gamma_i$. We consider the set $A_i(t)$ of all nodes that have received the token $t$ since $\gamma_0$: $A_{0}(t)=\emptyset$.

\begin{lemma}\label{tree}
$\restriction{t^{(\gamma_i)}.tab}{A_i(t)}$ represents a spanning tree of $(A_i(t), E\cap A_i(t)^2)$.
\end{lemma}

\begin{proof}
Obviously, $\restriction{t^{(\gamma_0)}.tab}{A_{0}(t)}=\emptyset$ is a spanning tree of $A_{0}(t)=\emptyset$.

The application of \textbf{R3} and \textbf{R4} entails no change on either $A_i$ or the token messages. Thus, we only consider the application of \textbf{R1} and \textbf{R2}

Consider a step $t^{(\gamma)}\rightarrow t^{(\gamma')}$ at which a node $i$ receives the token $t$ from $j$, and suppose that $\restriction{t^{(\gamma)}.tab}{A_{i(t)}}$ is a spanning tree of $A_i(t)$.

Then at the next step, $t^{(i+1)}.tab[j]=i$, $t^{(\gamma')}.tab[i]=i$ and $t^{(\gamma')}.tab[k]=t^{(\gamma)}.tab[k]$ for any other $k$ in $A_i(t)$. Since $\restriction{t^{(\gamma)}.tab}{A_i(t)}$ represents a spanning tree of $(A_i(t), V\cap A_i(t)^2)$, for any $k\neq i, j$, $(k, t^{(\gamma')}.tab[k])=(k, t^{(\gamma)}.tab[k])$ is an edge of $(A_i(t), V\cap A_i(t)^2)$ and $internal\_test$ will not remove $k$ from this array. Since $i$ has received the token from j, $(j, t^{(\gamma')}.tab[j])=(j, i)$ is an edge of $(A_{i+1}(t), V\cap A_{i+1}(t)^2)$.

If several tokens are pending reception by $j$, they may be merged: $t_1^{(\gamma)}\rightarrow t^{(\gamma')}$ and $t_2^{(\gamma)}\rightarrow t^{(\gamma')}$. Then, since $t_1$ and $t_2$ are correct, $t$ is also correct, and thus, $t^{(\gamma)}.tab$ is a spanning tree of $(A_i(t), E\cap A_i(t)^2)$.

Now two case can occur: either $i$ is in $A_i(t)$, or not. In both cases, setting $i$ as the root of the tree and the father of $j$, while leaving the remaining of the tree unchanged, gives a tree.

Thus, $\restriction{t^{(\gamma)}.tab}{A_{i}(t)}$ represents a spanning tree of $(A_i(t), E\cap A_i(t)^2)$.
\end{proof}

A node may belong to several spanning tree, if it has been visited by several tokens.

\begin{lemma}
The propagation of a reloading wave takes at most $n$ time units.
\end{lemma}
\begin{proof}
A time unit is the time taken by message sent to be received and treated. Now, the reloading wave is broadcast on a tree, of height at most $\mathcal N$. Thus, this propagation takes at most $\mathcal N$ time units.
\end{proof}

\begin{theorem}\label{nocreat}
A node in $A_i(t)$ cannot create a token.
\end{theorem}
\begin{proof}
Each time the token counter reaches $T_m-(n+1)$, a wave is propagated. The two lemma above guarantee that this wave hits any node in $A_i(t)$ in at most $n$ time units. Now, since it is in $A_i$, either this node has already received a reloading wave message, or it has received a token since the last wave was propagated. In both cases, it has reset its timeout to $T_m$ since the last wave initiation, \emph{ie} during the last $T_m-(n+1)$ time units. Thus, this timeout, at the initiation of the wave, is at least at $n+1$. Then, when the wave reaches the node, its timeout is $\geq1$, which makes it impossible for it to create a token between to successive waves, or between a token visit and the subsequent wave. Finally, no node in $A_i$ can create a token.
\end{proof}

Note that the cover property ensures that eventually, $A_i(t)=V$ \emph{whp}, so that:

\begin{lemma}\label{3}
$\mathcal A_3=\{\gamma\in \mathcal A_1\cap\mathcal A_2/\bigcup_i A_i(t)=V\}$ is a probabilistic attractor of $\mathcal A_1\cap\mathcal A_2$. In $\mathcal A_3$, rule \textbf{R2} can never be applied.
\end{lemma}


\subsubsection{There is eventually exactly one token}

The key assumption to verify this is that the meeting property of random walks hold: independent random numbers generators are needed. Also, when several tokens are headed to a same node, this node has to be able to detect it with probability $>0$, which is the case if the local treatment time is not negligible before the transmission time, or if messages are buffered for some non-negligible time before being treated.

\begin{definition}\label{LC}
A \emph{legitimate configuration} is a configuration with a single token $t$, the table of which represents a spanning tree of the system, and in which all nodes are hit by a reloading wave before their timers reach the value $0$.
$$\mathcal{LC} =\{\gamma \in \mathcal A_3/|Token_{\gamma}|=1\}$$
\end{definition}

The following theorem proves that $\LC$ matches the specification of $RandTokCirc$.

\begin{theorem}
A configuration $\gamma$ of $\LC$ is such that any execution $(\gamma_0=\gamma, \gamma_1, \gamma_2, \ldots)$ starting at $\gamma$ verifies
\begin{itemize}
\item $\forall k, |Token_{\gamma_k}|=1$;
\item $\forall k, \forall i, \exists l>k, |Token_{\gamma_{l}}(i)|=1$ \emph{whp}.
\end{itemize}
\end{theorem}

\begin{proof}
The closure property of $\LC$, that will be proved in the sequel, proves the first item. The second item comes from the fact that the successive positions of the random walk constitute a random walk, and thus verifies the hitting property.
\end{proof}

\begin{lemma}\label{4}
$\mathcal{LC}$ is a probabilistic attractor of $\mathcal{A}_3$
\end{lemma}

\begin{proof}
Consider a step $\gamma\rightarrow\gamma'$.

If $\gamma\rightarrow^{\mathrm{\mathbf{R3}}}\gamma'$ or  $\gamma\rightarrow^{\mathrm{\mathbf{R4}}}\gamma'$, $Token_{\gamma'}=Token_{\gamma}$. Now, according to theorem \ref{nocreat}, \textbf{R2} cannot be activated.

If $\gamma\rightarrow^{\mathrm{\mathbf{R1}}}\gamma'$, $Token_{\gamma'}=Token_\gamma\backslash T\cup\{t\}$, with $|T|\geq1$.

Thus, $1\leq|Token_{\gamma'}|\leq|Token_{\gamma}|$ (this is greater than 1 according to attractor $\mathcal A2$), which ensures closure of $\mathcal{LC}$. Now, if $|Token_\gamma|>1$, meeting property of random walks ensure that at some configuration $\gamma'\in E_\gamma$, several tokens are headed toward a same node. Then, if the treatment time is not negligible before the transmission time, another token is received with probability $>0$ during the treatment of the first token, and those token are merged. Thus, \emph{whp}, there is some configuration $\gamma''\in E_\gamma$ such that $|Token_{\gamma''}|<|Token_\gamma|$. 

Finally, eventually, a configuration $\delta$ is reached with $|Token_\delta|=1$.
\end{proof}

From Lemmas \ref{1}, \ref{2}, and \ref{4} 

\begin{theorem}[Convergence and closure]
The Algorithm, starting in an arbitrary configuration, converges to a configuration satisfying $\mathcal{LC}$ \emph{whp}.
\end{theorem}

\section{The impact of mobility}

The token circulation algorithm presented above is self-stabilizing. Thus, from any arbitrary configuration occurring because of a topological change, the algorithm eventually resumes its normal behavior if no further topological change occurs. If the time between two topological reconfigurations is greater than $K$ times the convergence time, then the system spends $\frac{K-1}K$ of the time in a correct configuration.

The token circulation itself is robust to topological changes, as shown in corollary \ref{RWProp}. However, we introduced mechanisms to ensure self-stabilization that can be affected by a topological change. Indeed, the reloading wave is based on a spanning tree computed in the course of the token circulation. This spanning tree can contain edges that have failed. In this case, the reloading wave cannot be propagated to all nodes. The timer of a node not receiving the reloading wave will then expire, leading to an undue token creation.

In this section, we study the probability that a topological change entails such an error. We also define a locally checkable criterion that ensures that no error occurs.

The only non-local topological information used by the algorithm is the spanning tree contained in the token and in the reloading wave messages. Thus, a topological modification has an impact only if it makes those trees inconsistent with the topology. The tree used in reloading wave messages is a subtree of the tree in the token at the time when the reloading wave is launched (algorithm \ref{Main}, rule R1.c).

Now, the tree in the token is updated each time the token hits a node (algorithm \ref{Main}, rule R1.a). After a topological change, a configuration is illegitimate if an edge that is in the token tree or in a reloading wave message is removed. This represents less than $2n-2$ edges in $m$. The walk of the token corrects the tree when the token hits the son of this edge in the tree. Thus, if no reloading wave is broadcast between the time at which the topological change occurs and the time at which the token hits this node, then the specification is met.

Thus, a single link disconnection has a probability $\frac{2n-2}m$ not to affect the algorithm. If the algorithm reaches an illegitimate configuration, it still has a probability $P[H_{ji}<T/2]$ to hit the son of the disconnected link before and correct the tree it contains before it launches a wave ($H_{ji}$ being the observed time, starting at node $j$ to reach node $j$). Thus, after a topological change, with probability $\frac{m-2n+2}mP[H_{ij}<T/2]$ (see the computations of this quantity in the next section), the algorithm continuously meets the specification.

An edge is in the tree if and only if it is the last link through which a node sent the token. If each node stores the link through which it sent the token last, the son in the tree of an link that has been disconnected can detect an illegitimate configuration. The configuration is illegitimate as long as the link through which a node sent the token is not present: from the link disconnection to the next visit of the token to the son of the link in the tree (see figure \ref{ar}: $i$ sends the token to $j$, that is its father until $i$ receives the token again).

Thus by replacing the statement in algorithm \ref{Main}:
\begin{algorithm}[H]
\begin{algorithmic}
\STATE Send $Token$ to $j$ chosen randomly in $N(i)$
\end{algorithmic}
\end{algorithm}

with:
\begin{algorithm}[H]
\begin{algorithmic}
\STATE Choose $j$ at random in $N(i)$
\STATE Send $Token$ to $j$
\STATE $father_i\rcp j$
\end{algorithmic}
\end{algorithm}

a wave propagation can be unsuccessful if and only if a node $i$ is such that $father_i\notin N_i$, which $i$ can detect.

\begin{figure}[H]
\begin{center}
\includegraphics[scale=0.45]{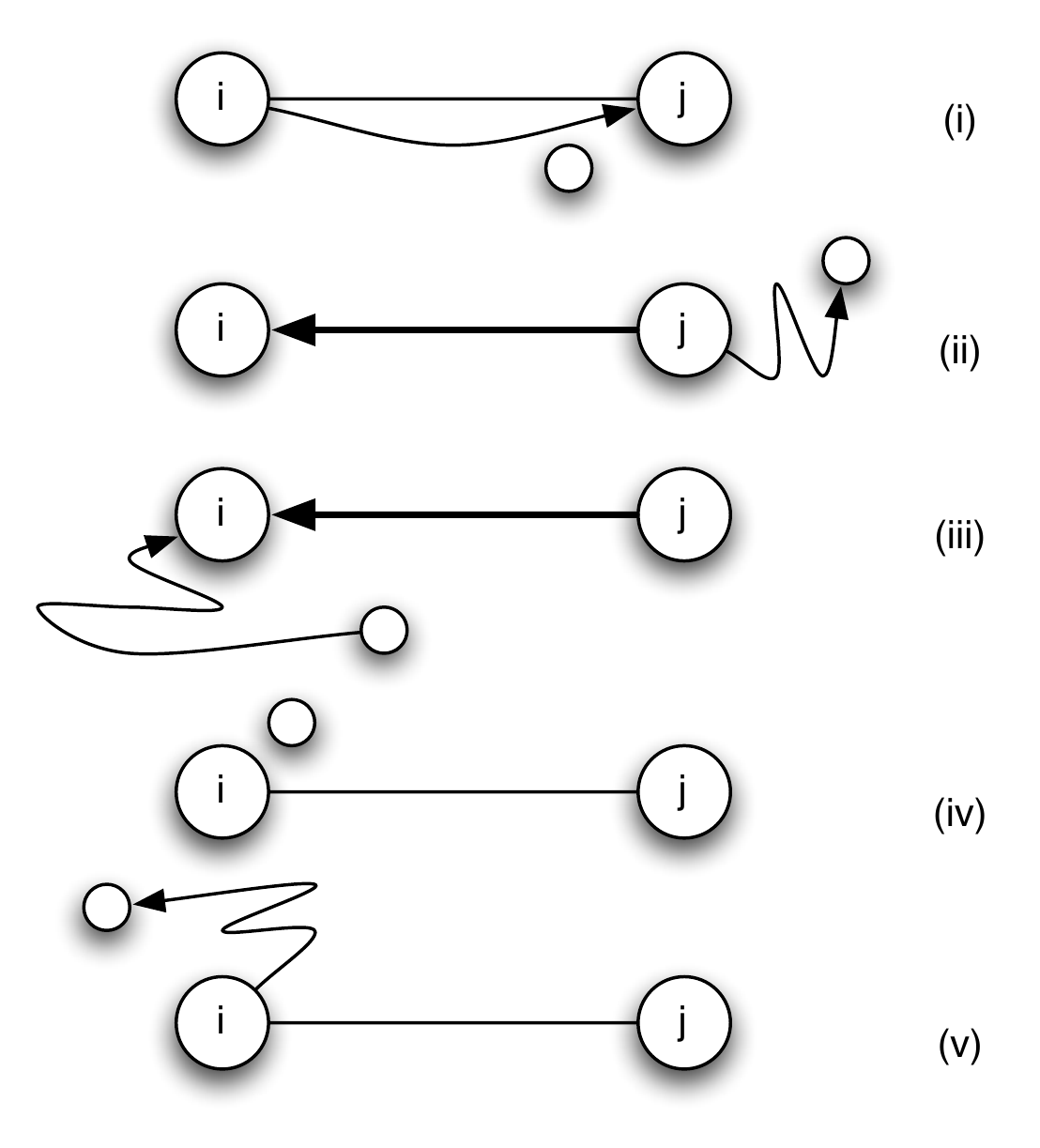}
\end{center}
\caption{Local detection of an illegitimate state}\label{ar} 
\end{figure}

Thus, we have:

\begin{property}
With probability $\frac{m-2n+2}mP[H_{ij}<T/2]$, after a link disconnection, the algorithm continuously respects the specification. If the system reaches an illegitimate configuration, a node in the system is aware of that.
\end{property}

\section{Timeout tuning}\label{Tt}

To solve the communication deadlock problem, the algorithm uses a decentralized timeout procedure: each processor indistinctly can produce a new token. To guarantee the stabilization property, a new mechanism, the reloading wave, is introduced. The role of this wave, periodically triggered,  is to prevent the creation of unnecessary tokens. 

Whatever the value proposed for $T_m$, as soon as this value is greater than $n$, the algorithm works correctly. But if $T_m$ is close to $n$, the reloading wave will be broadcasted too often, and if $T_m$ is too long, an absence of token will take a great amount of time before being corrected. We address in this section the problem to compute a good value for this timeout.

No bound can be given on the time a random walk takes to reach a given node (only results on expected times are available). However, as time goes by, it becomes improbable that the walk has not reached a node. 
 We first provide a probabilistic analysis of the waiting time. 
More precisely, we give a bound on the probability for a processor to 
wait for the token more than a certain amount of time. 
Then, we provide a criterion to decide a timeout value, based on the probability that the token is lost knowing that it has not been seen during a certain amount of time. This quantity depends on the probability that the token is lost during a transmission.

\subsection{Waiting times}

The waiting time is the average time a node is waiting for the token. It can be defined as the return time $h_{ii}$ ($=\frac{2m}{\deg(i)}$, see \cite{Lova93}), \emph{i.e.} the expected number of steps for the token, starting at node $i$, to return to node $i$ for the first time. 

It is interesting to measure the probability that a token has returned to a node after a given time. The probability that the token takes less than a given number of steps $t$ to come back to the node $i$ is defined by $P[H_{ii}\leq t]$ ($H_{ii}$ being the observed return time: $H_{ii}\geq t$ means that it has been more than $t$ steps since node $i$ has last seen the token). The following results give  a more accurate and comprehensive insight in the time a node will wait for the token after having released it.  In the sequel, we provide a bound on this value.

For the sake of simplicity, we will first study $P[H_{ii}\geq t]=1-P[H_{ii}\leq t+1]$.


\paragraph{Notation} Let  $\sigma [H_{ij}]$ denote the standard deviation of  $H_{ij}$ (the number of steps   to reach a node $j$ from $i$ for the first time), and $V[H_{ij}]$ the variance of $H_{ij}$.

The Chebyshev's inequality states that for any $\alpha$:
\begin{lemma}[Chebyshev's inequality]\label{Chebi}
$$P[H_{ii}\geq h_{ii}+ \alpha . \sigma [H_{ii}]  ]\leq \frac{1}{\alpha^2}$$
\end{lemma}

Thus we are led to compute the standard deviation of the hitting time. By definition, $\sigma [H_{ii}] = \sqrt{V[H_{ii}]}$ with $V[H_{ii}] = \mathbb{E}[(H_{ii} -h_{ii})^2]$. 

In the sequel, we present an algorithm to compute the variances of the return times on a graph, which is necessary to compute the Chebyshev bounds.

In order to compute the variance of the return time, we need to know the variances of all hitting times.
First, we state the following result:

\begin{lemma}[Variance of the number of steps to reach a node]\label{lemmevariance}
\begin{equation}\label{variance}
V[H_{ij}]+h_{ij}^2=\sum_{k\in\mathcal N(i)}p_{ik}(V[H_{kj}]+(h_{kj}+1)^2)
\end{equation}
\end{lemma}

\begin{proof}
$h_{ij}$ is the average length of the path a random walk starting from $i$ takes until it reaches $j$. Thus, since the probability that the random walk reaches $j$ is 1, the probability of an infinite random path not reaching $j$ is 0, and $$h_{ij}=\sum_{c\in\mathcal C_{i\rightarrow j}}p(c)l(c)$$
with $\mathcal C_{i\rightarrow j}$ the set of all paths from $i$ to $j$, $p(c)$, the probability that a random walk follows the path $c$ ($p(c)=\prod p_{c_kc_{k+1}}$), and $l(c)$ the length of $c$.


\begin{equation*}
\begin{split}
V[H_{ij}]&=\sum_{c\in\mathcal C_{i\rightarrow j}}p(c)(l(c)-h_{ij})^2\\
&=\sum_{c\in\mathcal C_{i\rightarrow j}}p(c)l(c)^2-h_{ij}^2\text{ according to a well-known identity}
\\
&=\sum_{c\in\mathcal C_{i\rightarrow j}}p_{c_0c_1}p(c_1c_2\ldots)(l(c_1c_2\ldots)+1)^2-h_{ij}^2\\
&
=\sum_{k\in\mathcal N(i)}p_{ik}\sum_{c\in\mathcal C_{k\rightarrow j}}p(c)(l(c)^2+2l(c)+1)-h_{ij}^2\\
&=\sum_{k\in\mathcal N(i)}p_{ik}(V[H_{kj}]+h_{kj}^2+2h_{kj}+1)-h_{ij}^2\\
&=\sum_{k\in\mathcal N(i)}p_{ik}(V[H_{kj}]+(h_{kj}+1)^2)-h_{ij}^2
\end{split}
\end{equation*}
\end{proof}

The system (\ref{variance}) is linear, and depends on the hitting times. In \cite{BuSo07}, we have proposed an efficient algorithm to compute the hitting times, with one matrix inversion.
In order to solve the system and obtain the variances, we have to compute the inverse of a matrix.

Let $M(j)$ the matrix defined by:\begin{itemize}
\item $M_{il}(j)=p_{il}=\frac{1}{\deg(i)}$ if $i\neq l$ and $i\neq j$ ;
\item $M_{ii}(j)=-1$ if $i\neq j$ ;
\item $M_{ji}(j)=0$ if $i\neq j$ ;
\item $M_{jj}(j)=1$.
\end{itemize}

Let $v(j)$ a vector defined by $v_i(j)=h_{ij}^2-\sum_{k\in\mathcal N(i)}p_{ik}(h_{kj}+1)^2$ for $i\neq j$ and $v_j(j)=0$, thus Lemma \ref{lemmevariance} can be rewritten: $$M(j)V[H_{.j}]=v(j)$$

$M(j)$ being inversible, we can compute the variances by finding its inverse.



From Lemma \ref{Chebi}, we have
\begin{corollary}\label{cor2}
Given a time $t$:
\begin{equation}\label{eq21}
P[H_{ii} < t]\geq1-\frac{V[H_{ii}]}{(t-h_{ii})^2}
\end{equation}

Given a probability $\varepsilon$:
\begin{equation}\label{eq22}
P\left[H_{ii} < h_{ii}+\frac{\sigma [H_{ij}]}{\sqrt{\varepsilon}}\right]\geq1-\varepsilon
\end{equation}
\end{corollary}

Expression (\ref{eq21}), provide a bound on the probability that the token has come back before a given time $t$. With (\ref{eq22}), we can have a time after which we are sure at a given confidence level $1-\epsilon$ that the token has come back.
\begin{figure}[H]
\begin{center}
\includegraphics[width=.3\linewidth, angle=270]{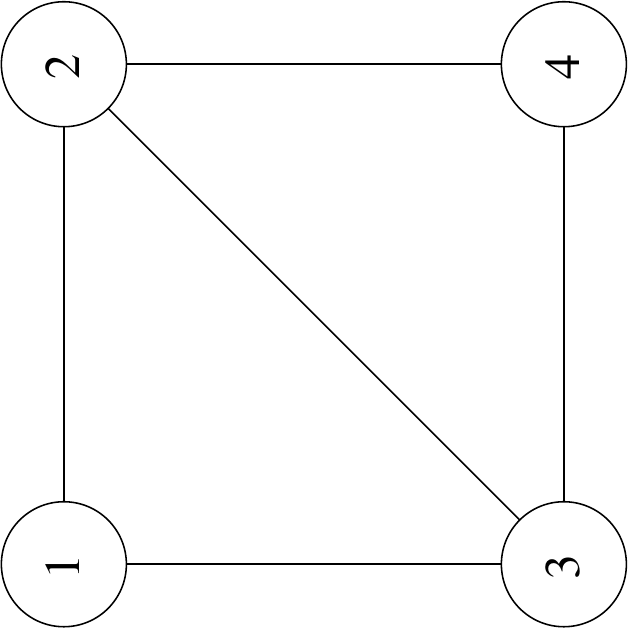}
\caption{Graph example $G$}\label{graph}
\end{center}
\end{figure}

To illustrate the meaning of the previous corollaries, consider the above graph. The return time is 5 for node 1. 

We use the corollary \ref{cor2} to obtain a good value for timeout. The variance of $H_{11}$ in the previous example is 51. Thus, for $t=50$, the probability that the node 1 waits less than 50 steps to receive the token after having released it is more than $1-\frac{V[H_{11}]}{t-h_{11}}=1-\frac{51}{45^2}\sim 97,5\%$. To be 99\% sure that the token has returned to 1, we will have to wait $h_{11}+10\times\sigma [H_{11}]\leq77$ steps.

\subsection{On timeout for deadlock communication}

We take into account possible transient failures which may remove the token from the network. In this subsection, we give a mean for the nodes to detect at any confidence level the loss of the token. 
We provide a way to choose the  best timeout value.

The longer a node has been waiting for the token, the more likely the token has disappeared. The suspicion that the token is lost increases with the time elapsed since it has seen the token for the last time. A node will have to check if the token has disappeared and then create a new token if necessary. 

We model the possibility that the token disappears by introducing a probability $p$ that the token disappears at each step: if the token exists at time $t$, at time $t+1$, the probability that it has disappeared is $p$ and the probability that it still exists is $1-p$.

\subsubsection{Measuring the probability that the token is lost}
We denote $L_t$ the event \emph{``at time $t$, the token is lost''}. We know, when the token cannot be lost, the probability that it comes back before a given time knowing that it still exists. We now want to compute the probability $P[L_t|H_{ii}>t]$ that the token is lost knowing that a node has not seen it in a given time.

\begin{theorem}
The probability that the token is lost, knowing that node $i$ has not seen it in the last $t$ steps, is:
$$P[L_t|H_{ii}>t]\geq1-\frac{V[H_{ii}]2^{t+1}(1-p)^{t+1}(1+p)}{2(1-p)^{t+1}t^2+pt^22^{t+1}}$$
where $V[H_{ii}]$ is the variance of $H_{ii}$, the number of steps before returning to  $i$ for the first time.\end{theorem}

\begin{proof}
Using the Bayes theorem, we obtain:
\begin{equation*}\begin{split}
P[L_t|H_{ii}>t]&=1-\frac{P[\neg L_t\cap\{H_{ii}>t\}]}{P[H_{ii}>t]}\\
&=1-\frac{P[\neg L_t\cap\{H_{ii}>t\}]P[\neg L_t]}{P[H_{ii}>t]P[\neg L_t]}\\
&=1-\frac{P[H_{ii}>t|\neg L_t]P[\neg L_t]}{P[H_{ii}>t]}
\end{split}\end{equation*}

According to the previous section, $P[H_{ii}>t|\neg L_t]\leq\frac{V[H_{ii}]}{t^2}$.

$\frac{1}{\min\{\deg\}^t} $ is a bound on the probability that the token  goes forth and back between two nodes during $t$ steps. 

We also have: 
$P[H_{ii}>t]\geq(1-p)^t\frac1{\min\{\deg\}^t}+p\sum_{k=1}^{t-1}(1-p)^k\frac1{\min\{\deg\}^k}=(1-p)^t\frac1{\min\{\deg\}^t}+p\frac{1-(1-p)^t\frac1{2^t}}{1-(1-p)\frac1{2}}\geq(1-p)^t\frac1{2^t}+p\frac{1-(1-p)^t\frac1{2^t}}{1-(1-p)\frac1{2}}$.

Thus,
\begin{small}
\begin{equation*}\begin{split}
P[L_t|H_{ii}>t]&\geq1-\frac{\frac{V[H_{ii}]}{t^2}(1-p)^{t+1}}{(1-p)^t\frac1{2^t}+p\frac{1-(1-p)^t\frac1{2^t}}{1-(1-p)\frac1{2}}}\\
&\geq1-\frac{\frac{V[H_{ii}]}{t^2}(1-p)^{t+1}(1-\frac{1-p}2)}{(1-p)^t\frac1{2^t}-(1-p)^{t+1}\frac1{2^{t+1}}+p(1-(1-p)^t\frac1{2^t})}\\
&\geq1-\frac{\frac{V[H_{ii}]}{t^2}(1-p)^{t+1}(1-\frac{1-p}2)}{(1-p)^{t+1}\frac1{2^t}-(1-p)^{t+1}\frac1{2^{t+1}}+p}\\
&\geq1-\frac{\frac{V[H_{ii}]}{t^2}(1-p)^{t+1}(1-\frac{1-p}2)}{(1-p)^{t+1}\frac1{2^{t+1}}+p}\\
&\geq1-\frac{V[H_{ii}]2^{t+1}(1-p)^{t+1}(1-\frac{1-p}2)}{(1-p)^{t+1}t^2+pt^22^{t+1}}\\
\end{split}\end{equation*}\end{small}
\end{proof}

\begin{figure}[H]
\begin{center}
\includegraphics[width=.8\linewidth]{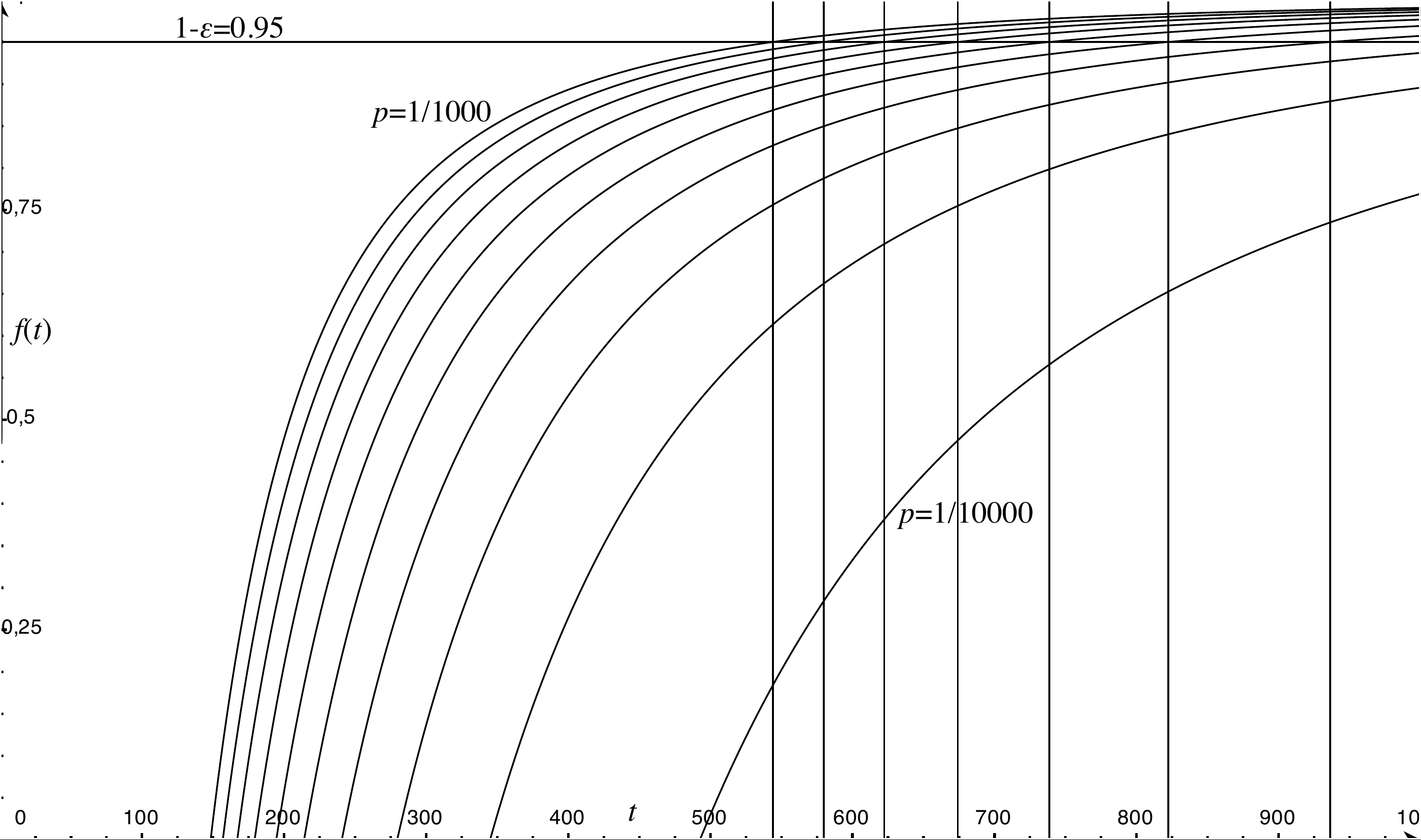}
\end{center}
\caption{Probability that the token is lost on graph $G$ according the elapsed time on a node and a probability $p$ that the token is lost during a step}
\end{figure}
The figure above represents the graph $\varepsilon=f(t)=\frac{V[H_{ii}]2^{t+1}(1-p)^{t+1}(1-\frac{1-p}2)}{(1-p)^{t+1}t^2+pt^22^{t+1}}$, with $V[H_{ii}]=51$,  and $p=0.1, 0.2,\ldots, 1$. To be $95\%$ sure that the token is lost, we look for the intersection of the curve with $1-\varepsilon =0.95$. If $p=0.1$, we can see that we will have to wait for 23 steps, if $p=0.5$, 38 steps, and if $p=0.9$, 75 steps.

\subsubsection{Choosing timeout values}
Choosing $t$ so that $$\frac{V[H_{ii}]2^{t+1}(1-p)^{t+1}(1-\frac{1-p}2)}{(1-p)^{t+1}t^2+pt^22^{t+1}}\leq\varepsilon$$ provides a time after which, if a node has not seen the token, the probability that it has disappeared is greater than $1-\varepsilon$.

\begin{theorem}[Timeout value]
Choosing a timeout greater than
$$\frac{\log \frac{V[H_{ii}]}{h_{ii}^2}+\log(\frac{1-p^2}{2p})-\log\varepsilon+2}{-\log(1-p)}$$
ensures that the token is lost with probability $1-\varepsilon$.
\end{theorem}

\begin{proof}

\begin{equation*}
\begin{split}
&\frac{V[H_{ii}]2^{t+1}(1-p)^{t+1}(1-\frac{1-p}2)}{(1-p)^{t+1}t^2+pt^22^{t+1}}\leq\varepsilon\\
\Leftrightarrow &\frac{(1-p)^{t+1}t^2+pt^22^{t+1}}{2^{t+1}(1-p)^{t+1}}\geq\frac{V[H_{ii}](1-\frac{1-p}2)}\varepsilon\\
\Leftrightarrow&\frac{t^2}{2^{t+1}}+\frac{pt^2}{(1-p)^{t+1}}\geq\frac{V[H_{ii}](1-\frac{1-p}2)}\varepsilon\\
\Leftrightarrow &\frac{pt^2}{(1-p)^{t+1}}\geq\frac{V[H_{ii}](1-\frac{1-p}2)}\varepsilon\\
\Leftrightarrow&\frac{t^2}{(1-p)^t}\geq\frac{V[H_{ii}](p+1)}{2\varepsilon p(1-p)}\\
\Leftrightarrow &2\log t-t\log(1-p)\geq\log \left(\frac{V[H_{ii}](p+1)}{2\varepsilon p(1-p)}\right)
\end{split}
\end{equation*}

When focusing only on $t\geq h_{ii}$: $\log t\geq\log h_{ii}$ and if $t$ is such that 
$$2\log h_{ii}-t\log(1-p)\geq\log C$$
then the probability that the token is lost is less than $1-\varepsilon$.

\end{proof}

In the above example, with $p=0.1$ and $\varepsilon=1\%$, we have to set the timeout to 33. With $\varepsilon=10\%$, the timeout is to be set at $23$.

\section{Conclusion}
We have proposed a self-stabilizing token circulation algorithm with no assumption on the topology of the distributed system. This algorithms can manage all events related to mobility, most of them without even requiring any convergence. The (average) convergence time is computed, and the trade-off between the number of messages and the convergence time is explained.

We now plan on working on the scalability of such solutions, with a quantitative assessment of the dynamicity of the considered systems.



\bibliographystyle{alpha}
\bibliography{biblio}

\newcommand{\etalchar}[1]{$^{#1}$}
\begin{thebibliography}{KLMT10}

\bibitem[AKL{\etalchar{+}}79]{AKLL+79}
R.~Aleliunas, R.~Karp, R.~Lipton, L.~Lovasz, and C.~Rackoff.
\newblock Random walks, universal traversal sequences and the complexity of
  maze problems.
\newblock In {\em 20th Annual Symposium on Foundations of Computer Science},
  pages 218--223, 1979.

\bibitem[BIZ89]{BaZe89}
Judit Bar-Ilan and Dror Zernik.
\newblock Random leaders and random spanning trees.
\newblock In {\em WDAG89}, pages 1--12. Springer-Verlag, 1989.

\bibitem[BS07]{BuSo07}
A.~Bui and D.~Sohier.
\newblock How to compute times of random walks based distributed algorithms.
\newblock {\em Fundamenta Informaticae, IOS Press}, 80(4):363--378, 2007.

\bibitem[Ciu10]{Ciuf10}
Augusto Ciuffoletti.
\newblock The wandering token: Congestion avoidance of a shared resource.
\newblock {\em Future Generation Computer Systems}, 26:473--478, 2010.

\bibitem[Coo11]{Coop11}
Colin Cooper.
\newblock Random walks, interacting particles, dynamic networks: Randomness can
  be helpful.
\newblock In {\em 18th International Colloquium on Structural Information and
  Communication Complexity, Gdansk, Poland, June, 2011}, volume 6796 of {\em
  Lecture Notes in Computer Science}, pages 1--14. Springer, 2011.

\bibitem[CW05]{ChWe05}
Yu~Chen and Jennifer~L. Welch.
\newblock Self-stabilizing dynamic mutual exclusion for mobile ad hoc networks.
\newblock {\em J. Parallel Distrib. Comput.}, 65(9):1072--1089, 2005.

\bibitem[Dij74]{Dijk74}
Edsger~W. Dijkstra.
\newblock Self-stabilizing systems in spite of distributed control.
\newblock {\em Commun. ACM}, 17(11):643--644, 1974.

\bibitem[Dol00]{Dole00}
Shlomi Dolev.
\newblock {\em Self-Stabilization}.
\newblock MIT Press, 2000.

\bibitem[DSW06]{DoSW06}
S.~Dolev, E.~Schiller, and J.~L. Welch.
\newblock Random walk for self-stabilizing group communication in ad hoc
  networks.
\newblock {\em IEEE Trans. Mob. Comput.}, 5(7):893--905, 2006.

\bibitem[Fei95a]{Feig95a}
Uriel Feige.
\newblock A tight lower bound on the cover time for random walks on graphs.
\newblock {\em Random Struct. Algorithms}, 6(4):433--438, 1995.

\bibitem[Fei95b]{Feig95}
Uriel Feige.
\newblock A tight upper bound on the cover time for random walks on graphs.
\newblock {\em Random Struct. Algorithms}, 6(1):51--54, 1995.

\bibitem[GM91]{GoMu91}
Mohamed~G. Gouda and Nicholas~J. Multari.
\newblock Stabilizing communication protocols.
\newblock {\em IEEE Trans. Computers}, 40(4):448--458, 1991.

\bibitem[HV01]{HaVi01}
Rachid Hadid and Vincent Villain.
\newblock A new efficient tool for the design of self-stabilizing l-exclusion
  algorithms: The controller.
\newblock In Ajoy~Kumar Datta and Ted Herman, editors, {\em WSS}, volume 2194
  of {\em Lecture Notes in Computer Science}, pages 136--151. Springer, 2001.

\bibitem[IJ90]{IsJa90}
Amos Israeli and Marc Jalfon.
\newblock Token management schemes and random walks yield self-stabilizing
  mutual exclusion.
\newblock In {\em PODC, ACM}, pages 119--131, 1990.

\bibitem[IKOY02]{IKOY02}
Satoshi Ikeda, Izumi Kubo, Norihiro Okumoto, and Masafumi Yamashita.
\newblock Fair circulation of a token.
\newblock {\em IEEE Trans. Parallel Distrib. Syst.}, 13(4):367--372, 2002.

\bibitem[KLMT10]{KLMT10}
Anne-Marie Kermarrec, Vincent Leroy, Afshin Moin, and Christopher Thraves.
\newblock Application of random walks to decentralized recommender systems.
\newblock In Chenyang Lu, Toshimitsu Masuzawa, and Mohamed Mosbah, editors,
  {\em Principles of Distributed Systems - 14th International Conference,
  OPODIS 2010, Tozeur, Tunisia, December 14-17, 2010}, volume 6490 of {\em
  Lecture Notes in Computer Science}, pages 48--63. Springer, 2010.

\bibitem[Lov93]{Lova93}
L.~Lov\'asz.
\newblock {Random walks on graphs : A Survey}.
\newblock In T.~Szonyi ed., D.~Miklos, and V.~T. Sos, editors, {\em
  Combinatorics : Paul Erdos is Eighty}, volume~2, pages 353--398. Janos Bolyai
  Mathematical Society, 1993.

\bibitem[NOSY10]{NOSY10}
Yoshiaki Nonaka, Hirotaka Ono, Kunihiko Sadakane, and Masafumi Yamashita.
\newblock Note: The hitting and cover times of metropolis walks.
\newblock {\em Theoretical Computer Science}, 411:1889--1894, March 2010.

\bibitem[TW91]{TeWi91}
Prasad Tetali and Peter Winkler.
\newblock On a random walk problem arising in self-stabilizing token
  management.
\newblock In {\em PODC}, pages 273--280, 1991.

\bibitem[Var00]{Varg00}
George Varghese.
\newblock Self-stabilization by counter flushing.
\newblock {\em SIAM J. Comput.}, 30(2):486--510, 2000.

\end{thebibliography}

\end{document}